\documentclass[conference]{IEEEtran}
\usepackage{booktabs} 
\usepackage{color}
\usepackage{graphicx}
\usepackage{amssymb}
\usepackage{amsthm}
\usepackage{amsmath}
\usepackage{amsfonts}
\usepackage[ruled,noend,linesnumbered]{algorithm2e}
\usepackage{textcomp}
\usepackage{textcomp}
\usepackage{url}%
\usepackage{enumerate}

\newlength{\figwidths}
\newlength{\expwidths}
\newlength{\expwidthd}
\newtheorem{example}{Example}
\newtheorem{theorem}{Theorem}
\newtheorem{proposition}[theorem]{Proposition}

\DeclareMathOperator*{\argmin}{arg\,min}

\usepackage[font=footnotesize,labelfont=bf]{caption}
\def\BibTeX{{\rm B\kern-.05em{\sc i\kern-.025em b}\kern-.08em
    T\kern-.1667em\lower.7ex\hbox{E}\kern-.125emX}}

\begin{document}

\title{Learning Individual Models for Imputation (Technical Report)}

\author{\IEEEauthorblockN{
Aoqian Zhang, 
Shaoxu Song, 
Yu Sun, 
Jianmin Wang}
\IEEEauthorblockA{
\textit{BNRist, Tsinghua University}, Beijing, China \\
\{zaq13, sxsong, sy17, jimwang\}@tsinghua.edu.cn}
}

\maketitle

\begin{abstract}
Missing numerical values are prevalent, 
e.g., owing to unreliable sensor reading, collection and transmission among heterogeneous sources.
Unlike categorized data imputation over a limited domain, 
the numerical values suffer from two issues: 
(1) \emph{sparsity problem}, the incomplete tuple may not have sufficient complete neighbors sharing the same/similar values for imputation, owing to the (almost) infinite domain; 
(2) \emph{heterogeneity problem}, 
different tuples may not fit the same (regression) model.
In this study, enlightened by the conditional dependencies that  hold conditionally over certain tuples rather than the whole relation,
we propose to learn a regression model individually for each complete tuple
together with its neighbors.
Our \textsf{IIM}, \emph{Imputation via Individual Models}, 
thus no longer relies on sharing similar values 
among the $\mathit{k}$ complete neighbors for imputation, 
but utilizes their regression results by the aforesaid learned individual (not necessary the same) models. 
Remarkably, we show that some existing methods are indeed special cases of our \textsf{IIM},
under the extreme settings of the number $\ell$ of learning neighbors considered in individual learning. 
In this sense, a proper number $\ell$ of neighbors is essential to learn the individual models  (avoid over-fitting or under-fitting). 
We propose to adaptively learn individual models over various number $\ell$ of neighbors for different complete tuples.
By devising efficient incremental computation, 
the time complexity of learning a model \emph{reduces from linear to constant}. 
Experiments on real data demonstrate that our \textsf{IIM} with adaptive learning achieves higher imputation 
accuracy than the existing approaches.
\end{abstract}

\section{Introduction}\label{sect:introduction}

Missing values are commonly observed \cite{DBLP:journals/tkde/0001SZLS16}, 
especially over numerical data \cite{DBLP:journals/dase/FirmaniMSB16}, for instance, owing to 
failures of sensor reading devices \cite{DBLP:conf/icde/JefferyAFHW06}, 
poorly handling overflow during calculation, 
mismatching in integrating heterogeneous sources \cite{DBLP:books/daglib/0029346}, 
and so on. 
Simply discarding the incomplete tuples with missing values makes the data even more incomplete. 

\begin{figure}[t]
\centering
\begin{minipage}{2.4in}
\includegraphics[width=0.6\figwidths]{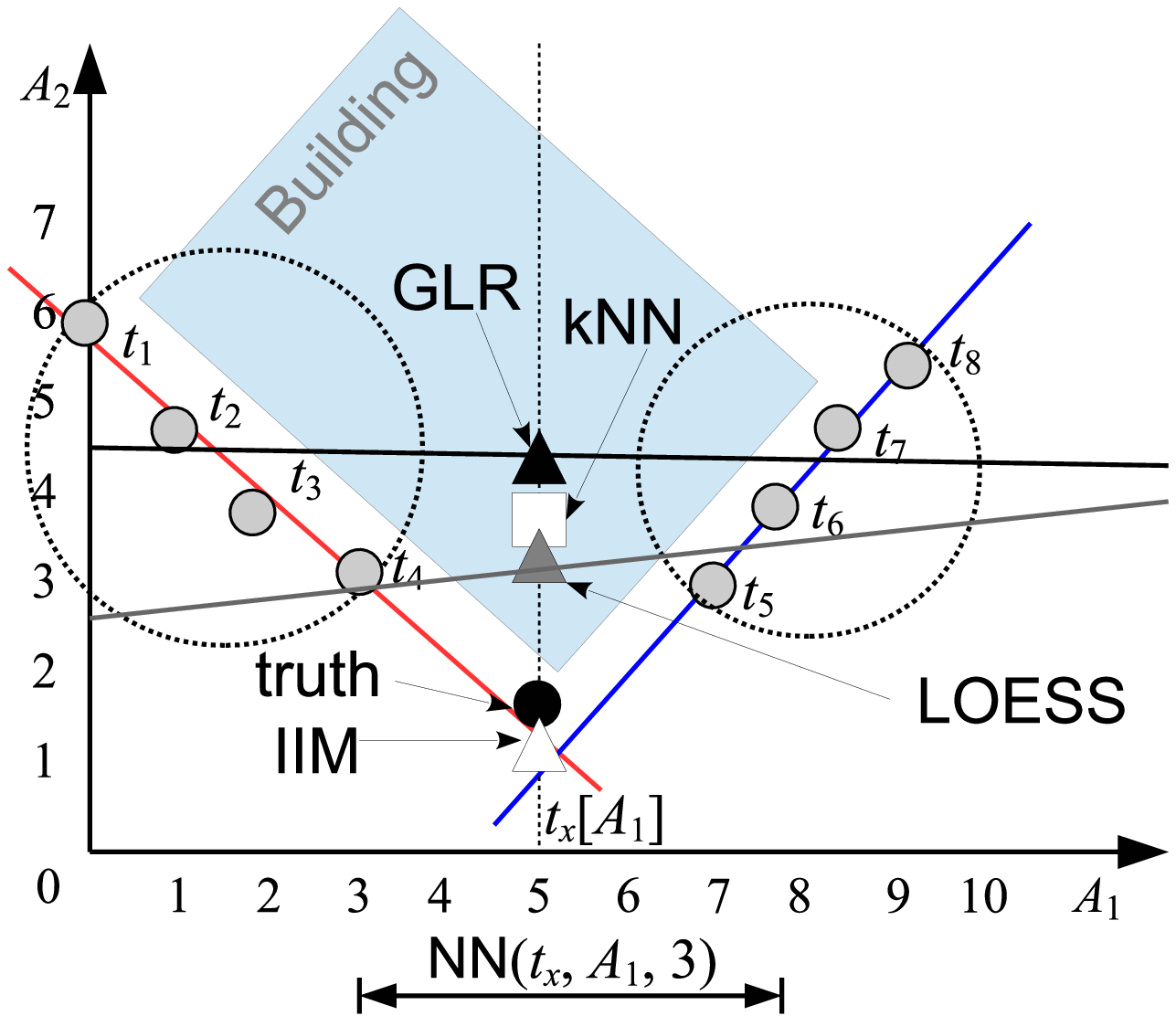} 
\end{minipage}
\begin{minipage}{1in}
\begin{footnotesize}
\begin{tabular}{lll}
 \hline\noalign{\smallskip}  Tid & $A_1$  & $A_2$  \\ \noalign{\smallskip}
 \hline\noalign{\smallskip}
 $\mathit{t}_1$ & $0$ & $5.8$ \\ \noalign{\smallskip}
 $\mathit{t}_2$ & $0.8$ & $4.6$ \\ \noalign{\smallskip}
 $\mathit{t}_3$ & $1.9$ & $3.8$ \\ \noalign{\smallskip}
 $\mathit{t}_4$ & $2.9$ & $3.2$ \\ \noalign{\smallskip}
 $\mathit{t}_5$ & $6.8$ & $3$  \\ \noalign{\smallskip}
 $\mathit{t}_6$ & $7.5$ & $4.1$ \\ \noalign{\smallskip}
 $\mathit{t}_7$ & $8.2$ & $4.8$ \\ \noalign{\smallskip}
 $\mathit{t}_8$ & $9$ & $5.5$ \\ \noalign{\smallskip}
 \hline\noalign{\smallskip}
 $\mathit{t}_\mathit{x}$ & $5$ & \textbf{--} \\ \noalign{\smallskip}
 \hline\noalign{\smallskip}
\end{tabular}
\end{footnotesize}
\end{minipage}
\caption{Motivation example of two-dimension data, where $\mathit{t}_\mathit{x}[\mathit{A}_2]$ is a missing value with ground truth 1.8. 
Our \textsf{IIM} learns the individual regression models (blue and red lines) w.r.t.\ heterogeneous neighbors ($\mathit{t}_4$ and $\mathit{t}_5$), 
instead of a same model (black or gray line) for all neighbors
}
\label{fig-example-correct}
\end{figure}

\subsection{Motivation}
\label{sect-motivation}

We notice that the existing imputation techniques \cite{little2014statistical,altman1992introduction,batista2002study} 
utilizing either complete attributes or complete tuples
suffer from two major issues, 
especially when handling numerical data from various sources. (See examples below.)

\subsubsection{Sparsity problem}
\label{sect-sparity-porblem}

The imputation via finding the closest complete tuple relies on the assumption that 
there exist neighbors sharing the same values.
Unfortunately, owing to the sparsity issue, such an assumption is often not the case in practice, 
e.g., $\mathit{t}_\mathit{x}$ in Figure \ref{fig-example-correct}
does not have any complete tuple sharing the same value. 
Thereby, the \textsf{kNN} method \cite{altman1992introduction} proposes to aggregate the values of complete neighbors. 

Owing to sparsity, 
there may not exist a complete tuple containing exactly the actual correct value of the incomplete tuple.
For this reason, it is also studied to impute a missing value from the regression model \cite{little1992regression}.
Instead of 
using a value directly from the complete tuple
(often unlikely to be the actual correct value owing to sparsity), 
the prediction based approach (\textsf{GLR}) \cite{little1992regression} assumes tuples sharing regression models.
For instance, in Figure \ref{fig-example-correct},
$\mathit{t}_5$ and $\mathit{t}_6$ have different values, but share the same regression model (blue line).

It is worth noting that 
even a complete tuple ($\mathit{t}_5$) is trusted (with no error), its value cannot be directly used as the imputation of the incomplete tuple ($\mathit{t}_\mathit{x}$), 
owing to the aforesaid sparsity issue.
However, this tuple $\mathit{t}_5$ can be used to learn a regression model. 
The incomplete tuple $\mathit{t}_\mathit{x}$ may not directly use the value of $\mathit{t}_5$, 
but use the value predicted by regression model of $\mathit{t}_5$, 
since neighbors may not share the same value but the regression model.

\subsubsection{Heterogeneity problem}
\label{sect-heterogeneity-porblem}

Since data often describe various facts or are collected from heterogeneous sources, 
no global semantics may fit the entire data \cite{DBLP:conf/icde/SahaS14}.
That is, there may not exist a single regression model that captures the semantics over all the data. 
Or generally speaking, \emph{one size does not fit all}. 
For instance, in Figure \ref{fig-example-correct}, 
a single \textsf{GLR} model (black line) cannot fit all the data points in different streets.

To address the heterogeneity problem, instead of assuming the same regression,
we argue to learn a fine-grained individual regression model that is only valid locally over a complete tuple and its neighbors.
For instance, in Figure \ref{fig-example-correct},
the individual regression model (red line) is only valid over $\mathit{t}_4$ and its neighbors such as $\mathit{t}_3$. 
Tuple $\mathit{t}_5$ in another street could have another regression model (blue line) that is distinct from $\mathit{t}_4$.
The imputation can thus utilize these more accurate individual models, 
instead of the imprecise global model (\textsf{GLR}) that does not fit all the data. 
The benefit of the imputation by individual models (\textsf{IIM}) would be the clearly higher accuracy than that of \textsf{GLR} with a single (inaccurate) global model, 
as the results shown in Table \ref{table-all-type}.

\begin{example}
Consider a check-in dataset
of two dimension in Figure \ref{fig-example-correct} for simplicity (more general, high dimensional data are considered in Section \ref{sect-experiment} of experiments).
Tuples $\mathit{t}_1-\mathit{t}_8$ (denoted by gray dots) represent 8 observations in the streets outside a building. 
There is another tuple $\mathit{t}_\mathit{x}$ with $\mathit{t}_\mathit{x}[\mathit{A}_1]=5$ observed but $\mathit{t}_\mathit{x}[\mathit{A}_2]$ missing during transmission (the truth of $\mathit{t}_\mathit{x}[\mathit{A}_2]$ is denoted by the black dot).

The nearest neighbor based imputation finds $\mathit{k}$ (say $\mathit{k}=3$) tuples that are most similar to $\mathit{t}_\mathit{x}$ on the complete attribute $\mathit{A}_1$, i.e., $\mathit{t}_4,\mathit{t}_5,\mathit{t}_6$. 
The mean value of three tuples on $\mathit{A}_2$ is then considered as the imputation of $\mathit{t}_\mathit{x}$ (\textsf{kNN}, white square).
Unfortunately, since no tuple is sufficiently close to the truth of $\mathit{t}_x$ (owing to sparsity), 
the imputation is not accurate.

The global linear regression (represented by solid black line) obviously cannot capture the difference between observations $\mathit{t}_1-\mathit{t}_4$ and $\mathit{t}_5-\mathit{t}_8$ in two streets. 
The imputation by the global regression (\textsf{GLR}, black triangle) is not accurate. 

The local regression assumes a same regression locally over the neighbors $\mathit{t}_4,\mathit{t}_5,\mathit{t}_6$ of the incomplete tuple $\mathit{t}_\mathit{x}$, found on the complete attribute $\mathit{A}_1$. 
Again, owing to the heterogeneity issue, $\mathit{t}_5,\mathit{t}_6$ and $\mathit{t}_4$ from two streets, respectively, indeed have different regression models. 
The imputation by the local regression (\textsf{LOESS}, gray triangle) is not accurate either. 
\end{example}

The idea of \textsf{IIM} is enlightened by the conditional dependencies \cite{DBLP:conf/icde/BohannonFGJK07}, 
which only hold conditionally over certain tuples rather than the whole relation. 
That is, the constraint does not fit all the data, but only applies to a subset of tuples specified by certain conditions. 
Analogously, a regression model may not fit all the data, but only applies ``conditionally'' to the nearby neighbors of a tuple. 
Thereby, we propose to learn a regression model individually for each complete tuple and its neighbors, 
instead of a single global regression model that cannot fit all the tuples.

\subsection{Proposal}
\label{sect-proposal}

The Imputation via Individual Models (\textsf{IIM}) proposed in this paper thus has two phases: 
(1) the \emph{learning phase} learns individually a regression for each complete tuple
together with its neighbors,
e.g., $\mathit{f}_1,\dots,\mathit{f}_3$ for $\mathit{t}_1,\dots,\mathit{t}_3$, 
respectively, in Figure \ref{fig:model}; and
(2) the \emph{imputation phase} finds $\mathit{k}$ complete imputation neighbors of the incomplete tuple, 
and aggregate the regression results produced by the aforesaid learned individual regression models of the $\mathit{k}$ complete neighbors. 

For example, 
$\mathit{t}_\mathit{x}$ could use the regression models of neighbors $\mathit{t}_4, \mathit{t}_5$ and $\mathit{t}_6$,
and aggregate the results of different regressions as the imputation (\textsf{IIM}, white triangle in Figure~\ref{fig-example-correct}).

A key issue is how to perform individual learning for each complete tuple. 
To learn the individual model,
it needs to find a 
number of $\ell$ learning neighbors that are similar to the tuple.
A different number of learning neighbors lead to various 
learned models.
Determining the number $\ell$ of neighbors for learning
is highly non-trivial. 
For each complete tuple, 
(1) if the number $\ell$ is too small, the learned regression model may \emph{overfit} the data; 
(2) on the other hand, if $\ell$ is too large (e.g., considering almost all the heterogeneous tuples in the dataset like global regression), 
it leads to \emph{under-fitting}.
(We address the overfitting and under-fitting issues by adaptive learning below.)

\subsection{Contribution}
\label{sect-contribution}

Our major contributions in this paper are summarized as:

(1) We propose a novel approach \textsf{IIM} of Imputation via Individual Models (Section \ref{sect-framework}), 
with learning and imputation phases as aforesaid.
The heterogeneity issue is addressed by learning an individual model for each tuple
together with its neighbors.
\textsf{IIM} does not directly use the values of complete neighbors for imputation (but their models) and thus tackles the sparsity problem. 

(2)
We prove that some existing approaches are indeed the special cases of \textsf{IIM} under extreme settings 
(i.e., $\ell=1$ or $\ell=\mathit{n}$ in Propositions \ref{the:l0-equal} and \ref{the-LR-equal} in Section \ref{sect-subsume}). 
It does not only illustrate the rationale of our proposal, but also motivate us to adaptively determine a proper $\ell$ (in between the extreme $1$ and $\mathit{n}$) for each tuple to avoid over-fitting or under-fitting.

(3) We adaptively learn the individual model for each complete tuple over a distinct number $\ell$ of learning neighbors (Section \ref{sect-adaptive}).  
By introducing a validation step, 
we determine a proper number $\ell$ and the corresponding learned model for each  complete tuple, 
which can impute most accurately the other complete tuples (considered as validation set). 
Experiments show that the adaptively learned individual models indeed lead to better imputation results.
Efficient incremental computation is devised for adaptive learning, 
which reduces the time complexity of learning a model 
from linear to constant.

(4) We conduct extensive experiments over real datasets (Section \ref{sect-experiment}).  
The results demonstrate that our \textsf{IIM}  has  significantly better performance than the state-of-the-art imputation methods.
Remarkably, we show that the proposed imputation indeed 
improves the accuracy of  classification application over the data with 
\emph{%
real-world missing values. 
}

Table \ref{table-notations} 
lists the frequently used notations. 
 \begin{table}[t]
 \caption{Notations}
 \label{table-notations}
 \centering
 \begin{tabular}{rp{3in}}
 \hline\noalign{\smallskip}
 $\mathcal{R}$ & schema on $\mathit{m}$ attributes \\ \noalign{\smallskip}
 $\mathit{r}$ & relation of $\mathit{n}$ complete tuples \\ \noalign{\smallskip}
 $\mathit{t}_\mathit{x}$ & incomplete tuple with missing value  \\ \noalign{\smallskip}
 $\mathit{A}_\mathit{x}$ & incomplete attribute in $\mathit{t}_\mathit{x}$, $\mathit{A}_\mathit{m}$ by default for simplicity \\ \noalign{\smallskip}
 $\mathcal{F}$ & complete attributes in $\mathit{t}_\mathit{x}$, $\mathcal{R}\setminus\{\mathit{A}_\mathit{m}\}$ by default \\ \noalign{\smallskip}
 $\phi$ & parameter of linear regression model \\ \noalign{\smallskip}
 $\ell$ & number of learning neighbors for learning the individual model of a complete tuple \\ \noalign{\smallskip}
 $\mathit{k}$ & number of imputation neighbors for imputing an incomplete tuple \\ \noalign{\smallskip}
 \hline
 \end{tabular}
 \end{table}

\section{Preliminary and Related Work}\label{sect:preliminary}

In this section, we introduce preliminaries and categorize major imputation approaches into two classes in Table \ref{table-algorithm}.
The key ideas of imputation based on tuple models and attribute models are presented in Figure \ref{fig:model}. 
We discuss that each category of existing techniques suffers from either  the heterogeneity or the sparsity problem.
It motivates us to devise the novel imputation via individual models in Section~\ref{sect-framework}.

Consider a relation $\mathit{r}$ of $\mathit{n}$ tuples 
$\mathit{r}=\{\mathit{t}_1, \mathit{t}_2, \ldots, \mathit{t}_\mathit{n}\}$, 
with schema 
$\mathcal{R} = \{\mathit{A}_1, \mathit{A}_2, \ldots \mathit{A}_\mathit{m} \}$
on $\mathit{m}$ attributes.
We denote $\mathit{t}_i[A_j]$ the value of tuple $\mathit{t}_{i}\in\mathit{r}$ on attribute $A_{j}\in\mathcal{R}$.

Let $\mathit{t}_\mathit{x}$ be a tuple over $\mathcal{R}$ with missing value on attribute $\mathit{A}_\mathit{x}$. 
We call $\mathit{A}_\mathit{x}$ the incomplete attribute and 
$\mathcal{F}=\mathcal{R}\setminus\{\mathit{A}_\mathit{x}\}$ the complete attributes.
(For simplicity, we consider $\mathit{A}_\mathit{m}$ as the incomplete attribute by default. 
Missing values on other attributes could be addressed similarly.
Multiple incomplete attributes in a tuple could be addressed one by one.)

\begin{table}[t]
 \caption{Imputation methods considered in (empirical) comparison}
 \label{table-algorithm}
 \centering
 \begin{tabular}{llp{1.6in}}
 \hline\noalign{\smallskip} Approach & Model & Property \\ \noalign{\smallskip}
 \hline\noalign{\smallskip}
 \textsf{Mean} \cite{DBLP:journals/tsmc/FarhangfarKP07} & Tuple & Global average \\ \noalign{\smallskip}
 \textsf{kNN} \cite{altman1992introduction} & Tuple & Local average\\ \noalign{\smallskip}
 \textsf{kNNE} \cite{DBLP:conf/icpr/DomeniconiY04} & Tuple & kNN Ensemble \\ \noalign{\smallskip}
 \textsf{IFC} \cite{DBLP:conf/fuzzIEEE/NikfalazarYBK17} & Tuple & Cluster average \\ \noalign{\smallskip}
 \textsf{GMM} \cite{yan2015missing} & Tuple & Cluster average \\ \noalign{\smallskip}
 \textsf{SVD} \cite{DBLP:journals/bioinformatics/TroyanskayaCSBHTBA01} & Tuple & $k$ most significant eigengenes \\ \noalign{\smallskip}
 \textsf{ILLS} \cite{DBLP:conf/apbc/CaiHL06} & Tuple & Local regression over tuples  \\ \noalign{\smallskip}
 \hline\noalign{\smallskip}
 \textsf{GLR} \cite{little1992regression} & Attribute & Global regression \\ \noalign{\smallskip}
 \textsf{LOESS} \cite{cleveland1996smoothing} & Attribute & Local regression  \\ \noalign{\smallskip}
 \textsf{BLR} \cite{rubin2004multiple} & Attribute & Bayesian linear regression \\ \noalign{\smallskip}
 \textsf{ERACER} \cite{DBLP:conf/sigmod/MayfieldNP10} & Attribute & Neighbor regression  \\ \noalign{\smallskip}
 \textsf{PMM} \cite{landerman1997empirical} & Attribute & Predictive mean matching \\ \noalign{\smallskip}
 \textsf{XGB} \cite{DBLP:conf/kdd/ChenG16} & Attribute & Xgboost, tree boosting system \\ \noalign{\smallskip}
 \hline\noalign{\smallskip}
 \end{tabular}
\end{table}

\subsection{Imputation based on Tuple Models}
\label{sect-related-tuple}

\subsubsection{Nearest Neighbor Model \textsf{kNN}}
\label{sect-existing-kNN}

To impute the missing numerical values, a natural idea is
to retrieve similar complete instances from $\mathit{r}$ for imputation, 
known as the \emph{$\mathit{k}$-nearest-neighbor} approach, \textsf{kNN} \cite{altman1992introduction, batista2002study}.

Let
$\textsf{NN}(\mathit{t}_\mathit{x}, \mathcal{F}, \mathit{k})$ 
be  $\mathit{k}$ nearest neighbors of $\mathit{t}_\mathit{x}$ on attributes $\mathcal{F}$ from $\mathit{r}$, 
e.g., with the smallest Euclidean distance \cite{DBLP:conf/kdd/AnagnostopoulosT14}
\begin{align}\label{equation-distance}
\textstyle
\mathit{d}_{x,i} = 
\sqrt{
\frac{
\sum\limits_{\mathit{A}\in\mathcal{F}}(\mathit{t}_{x}[\mathit{A}] - \mathit{t}_{i}[\mathit{A}])^2
}{
|\mathcal{F}|
}
}
\end{align}
where $\mathit{d}_{x,i}$ denotes the distance between tuple $\mathit{t}_\mathit{x}$ and $\mathit{t}_i$ on complete attributes $\mathcal{F}$.

The $\textsf{kNN}$ imputation is in two steps:
(1) find $k$ nearest neighbors $\mathit{T}_\mathit{x} = \textsf{NN}(\mathit{t}_\mathit{x}, \mathcal{F}, \mathit{k})$, 
and 
(2) use the $\mathit{A}_\mathit{m}$ values of neighbors for imputation, 
e.g., by arithmetic mean 
\begin{align}\label{equation-kNN}
\textstyle
\mathit{t}_\mathit{x}'[\mathit{A}_\mathit{m}] = \frac{
\sum\limits_{\mathit{t}_j \in \mathit{T}_\mathit{x}}\mathit{t}_j[A_m]
}{
\mathit{k}
}.
\end{align}

\begin{figure}[t]
\centering
\includegraphics[width=0.8\figwidths]{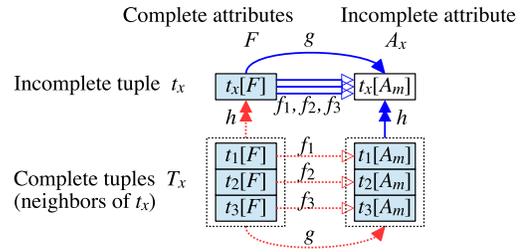}
\caption{Learning (dashed arrows) models over complete data and imputing (solid arrows) the missing value $\mathit{t}_\mathit{x}[\mathit{A}_\mathit{m}]$ by tuple model $h$, attribute model $g$, or individual models $f_1,\dots,f_3$ w.r.t.\ $\mathit{t}_1,\dots,\mathit{t}_3$}
\label{fig:model}
\end{figure}

\subsubsection{Variations of Tuple Models}

The first variation is on the neighbors in step (1) of the \textsf{kNN} imputation. 
\textsf{kNNE} \cite{DBLP:conf/icpr/DomeniconiY04} finds different groups of $\mathit{k}$ neighbors by computing distances on various subsets of features and then combine the imputation results from these different groups.
Instead of $\mathit{k}$ neighbors, 
the \textsf{Mean} method \cite{DBLP:journals/tsmc/FarhangfarKP07}  simply identifies all the tuples (as $\mathit{T}_\mathit{x}$) for aggregation in the following step.
Clustering is also employed to identify the neighbors for imputation, 
e.g.,  \textsf{IFC} \cite{DBLP:conf/fuzzIEEE/NikfalazarYBK17} considering fuzzy k-means \cite{DBLP:conf/rsctc/LiDSS04} 
or \textsf{GMM} \cite{yan2015missing} using the Gaussian mixture model. 
Similarity rules \cite{DBLP:journals/tkde/Song0C14,DBLP:journals/tods/Song011} are also employed to identify the neighbors \cite{DBLP:journals/pvldb/SongZC015,DBLP:journals/tkde/SongSZCW20}. 
Moreover, instead of searching existing data as neighbors, 
the \textsf{SVD} approach \cite{DBLP:journals/bioinformatics/TroyanskayaCSBHTBA01} finds a set of mutually orthogonal expression patterns 
(so-called eigenvectors)
as $\mathit{T}_\mathit{x}$ for aggregation imputation.

The second variation is on the aggregation model in step (2) of the \textsf{kNN} imputation. 
In addition to Formula \ref{equation-kNN}, 
more advanced aggregation considers the distances of neighbors as aggregation weights \cite{DBLP:conf/kdd/AnagnostopoulosT14}.
Furthermore, instead of the model of aggregating $\mathit{t}_j[A_m]$ over $\mathit{t}_j \in \mathit{T}_\mathit{x}$, 
\textsf{ILLS} \cite{DBLP:conf/apbc/CaiHL06} learns a model $\mathit{h}$ for predicting $\mathit{t}_\mathit{x}$ values from $\mathit{T}_\mathit{x}$.
In this sense, the arithmetic mean aggregation in Formula \ref{equation-kNN} is a special $\mathit{h}$ that does not need learning from $\mathit{T}_\mathit{x}[\mathcal{F}]$ and $\mathit{t}_\mathit{x}[\mathcal{F}]$.
We call $\mathit{h}$ a tuple model, and this category the tuple model-based imputation.

\subsubsection{Discussion}
\label{sect-knn-discussion}

The idea of learning over individual tuples and their closest neighbors in our proposal \textsf{IIM} 
is related to past work \textsf{kNN} \cite{altman1992introduction}. 
The difference is that 
to impute the incomplete tuple $\mathit{t}_\mathit{x}$, 
\textsf{kNN} uses (aggregates) directly the values of the $\mathit{k}$-closest neighbors $\mathit{t}_i$ of $\mathit{t}_\mathit{x}$ as the imputation, 
while our \textsf{IIM} learns individual models for the neighbor tuples $\mathit{t}_i$ by considering their $\ell$-closest neighbors $\mathit{t}_j$, respectively. 
The values predicted by the learned models are then aggregated as the imputation.
The defeat of directly using the values of $k$-closest neighbors to impute missing values is that 
owing to sparsity, no sufficient neighbors could be found sharing similar values with incomplete tuple $\mathit{t}_\mathit{x}$. 
For instance, $\mathit{t}_\mathit{x}$ in Figure \ref{fig-example-correct} does not have any tuple sharing highly similar values.
Alternatively, 
we learn a model from the tuple and its $\ell$-closest neighbors. 
Tuples may not share the same values but models. 
For example, $\mathit{t}_\mathit{x}$ in Figure \ref{fig-example-correct} fits the model that is learned from $\mathit{t}_4$ and its neighbors, 
and is thus accurately imputed.

\subsection{Imputation based on Attribute Models}
\label{sect-related-attribute}

\subsubsection{Linear Regression Model \textsf{GLR}}
\label{sect-existing-reg}

Rather than capturing relationships to the complete tuples in $\mathit{r}$, 
another well-known idea is 
to explore the relationships between incomplete and complete attributes, e.g., by the linear regression model \cite{little2014statistical}.

Let $\textsf{LR}(\mathcal{F},\mathit{A}_\mathit{m}, \mathcal{R})$
denote the linear regression model from complete attributes $\mathcal{F}$ to incomplete attribute $\mathit{A}_\mathit{m}$, having
\begin{align} \nonumber
\mathit{t}[A_m] &= \phi[C]1+\phi[A_1]\mathit{t}[A_1] + \ldots + \phi[A_{m-1}]\mathit{t}[A_{m-1}] + \varepsilon \\ 
\label{equation-regression}
  &= (1,\mathit{t}[\mathcal{F}]) \phi + \varepsilon
\end{align}
where $\mathit{t}$ is a tuple over $\mathcal{R}$, 
$\phi=\{\phi[C], \phi[A_1], \ldots, \allowbreak \phi[A_{m-1}]\}^\top$ is the parameter  of linear regression ($\phi[C]$ denotes the constant term),
and
$\varepsilon$ is the error term.

The imputation is thus in two steps:
(1) learn parameter $\phi$ from relation $\mathit{r}$ of complete tuples (e.g., by Ordinary Least Square or Ridge Regression \cite{rao1973linear}, see more details in Section \ref{sect-learning}),
and 
(2) perform the imputation referring to the learned linear regression model, 
\begin{align}\label{equation-imputation-global}
\mathit{t}_\mathit{x}'[A_m] = (1, \mathit{t}_\mathit{x}[\mathcal{F}])\phi.
\end{align}
Since the linear regression is declared on all tuples over $\mathcal{R}$, 
we call this \emph{global linear regression} method, \textsf{GLR}.

\subsubsection{Variations on Attribute Models}

Similar to the idea of aggregating only \textsf{kNN} tuples \cite{altman1992introduction}  
rather than \textsf{Mean} \cite{DBLP:journals/tsmc/FarhangfarKP07} of all tuples in Section \ref{sect-related-tuple}, 
\textsf{LOESS} \cite{cleveland1996smoothing} learns a local regression over the neighbors 
$\textsf{NN}(\mathit{t}_\mathit{x}, \mathcal{F}, \mathit{k})$ 
of $\mathit{t}_\mathit{x}$, 
instead of the global regression over all tuples. 
Statistical analysis could be further employed to linear regression, 
e.g., Bayesian linear regression \textsf{BLR} in the context of Bayesian inference.
(We use the MICE \cite{buuren2010mice} implementation mice.norm in R in experiments.)

The attribute models can cooperate with the tuple models.
The \textsf{ERACER} approach \cite{DBLP:conf/sigmod/MayfieldNP10} further studies the regression model over neighbors, i.e., combining both $g$ and $h$ in Figure \ref{fig:model}.
For instance, the temperature of a sensor is related to its humidity ($g$), 
as well as its neighbors' temperature and humidity ($h$).
The predictive mean matching \textsf{PMM} \cite{landerman1997empirical} does not directly use the value $\mathit{t}_\mathit{x}'[A_m]$ predicated by linear regression as the imputation. 
Instead, it finds neighbors whose predications also by the same linear regression are most similar to the predicated value $\mathit{t}_\mathit{x}'[A_m]$. 
A randomly selected original value $\mathit{t}_j[A_m]$ of the identified neighbors $\mathit{t}_j$ is returned as the imputation. 
The widely used XGboost \cite{DBLP:conf/kdd/ChenG16} (\textsf{XGB}) algorithm learns a set of classification and regression trees and ensembles the results. 
(We use the MICE \cite{buuren2010mice} implementation mice.pmm and library `xgboost' in R in the experiments.)

\subsubsection{Discussion}

Owing to the heterogeneity problem, 
assuming the same regression either globally, locally or randomly (for xgboost) \cite{cleveland1996smoothing} for different tuples could be indefensible.  

\section{Imputation via Individual Learning}\label{sect-framework}

As illustrated in Figure \ref{fig:model}, 
the Imputation via Individual Models (\textsf{IIM}) addresses the heterogeneity and sparsity problems in two aspects, respectively.
(1) 
The learning phase in Section \ref{sect-learning}
learns a linear regression model individually for each tuple 
(together with its neighbors,
e.g., models $f_1, \dots, f_3$ in Figure \ref{fig:model}), 
instead of assuming the same regression for different tuples (with heterogeneity). 
This is enlightened by the conditional dependencies that  hold conditionally over certain tuples \cite{DBLP:conf/icde/BohannonFGJK07}.
(2) 
The imputation phase in Section~\ref{sect-imputation}
aggregates the regression results of multiple individual regression models suggested by different neighbors, 
rather than relying the neighbors to have similar values (suffering sparsity).

\subsection{Learning Phase}
\label{sect-learning}

The \emph{learning phase} learns the parameter $\phi_i$ of the linear regression model (in Formula \ref{equation-regression}) individually for each tuple $\mathit{t}_i\in\mathit{r}$.
The learned individual regression models are then utilized in the imputation in Section \ref{sect-imputation}.

\begin{algorithm}[t]
\caption{Learning($\mathit{r}, \ell$, $\mathcal{F}$, $\mathit{A}_\mathit{m}$)}
\label{algorithm-learning}
 \KwIn{relation $\mathit{r}$ of complete tuples, number $\ell$ of learning neighbors, complete attributes $\mathcal{F}$, incomplete attribute $\mathit{A}_\mathit{m}$}
 \KwOut{$\Phi$ the set of regression parameters $\phi_i$ learned for all tuples $\mathit{t}_i$ in $\mathit{r}$}

\For{each $\mathit{t}_i \in \mathit{r}$}{\label{line-learning-tuple}
  $\mathit{T}_i \leftarrow \textsf{NN}(\mathit{t}_i, \mathcal{F}, \ell)$\; \label{line-learning-lNN} 
  $\phi_i \leftarrow \textsf{LR}(\mathcal{F}, \mathit{A}_\mathit{m}, \mathit{T}_i)$\;
  \label{line-learning-parameter}
}
\Return{$\Phi$ }
\end{algorithm}

Algorithm \ref{algorithm-learning} presents the procedure of individual learning over $\mathit{r}$ for the regression from $\mathcal{F}$ to $\mathit{A}_\mathit{m}$.
For each  $\mathit{t}_i \in \mathit{r}$, 
we consider a set $\mathit{T}_i$ of nearest neighbors
i.e.,
$\textsf{NN}(\mathit{t}_i, \mathcal{F}, \ell)$
in Line \ref{line-learning-lNN},
a.k.a. learning neighbors.
They
are obtained in the same way of 
obtaining $\mathit{k}$ nearest neighbors in the \textsf{kNN} approach, 
$\textsf{NN}(\mathit{t}_\mathit{x}, \mathcal{F}, \mathit{k})$, 
as introduced in Section \ref{sect-existing-kNN}.  
That is, return the tuples with the smallest Euclidean distance on attributes $\mathcal{F}$ \cite{DBLP:conf/kdd/AnagnostopoulosT14}.
In case of sparsity, 
the returned neighbors may not share similar values with the incomplete tuple, 
and thus the \textsf{kNN} approach
directly aggregating the values of nearest neighbors is not accurate. 
To deal with sparsity, we propose to learn regression models over the nearest neighbors, 
and use the learned models to predict the missing value.

Let $\ell$ be the number of $\mathit{t}_i$'s neighbors considered in individual learning, 
namely the number of learning neighbors.
As stated in Section \ref{sect-motivation}, 
the number $\ell$ should be sufficiently large to avoid overfitting, 
but not too large owing to heterogeneity. 
A straightforward idea is to simply consider a fixed number $\ell$ for all the tuples in $\mathit{r}$
(see Section \ref{sect-experiment-adaptive} for empirical results on considering various fixed $\ell$).
More advanced adaptive learning considering distinct number of learning neighbors for various tuples in $\mathit{r}$ is devised in Section \ref{sect-adaptive}.

\subsubsection{Learning Regression Parameter}

Given a set of tuples,  
$\mathit{T}_i=\{\mathit{t}_1, \mathit{t}_2, \ldots, \mathit{t}_\ell\} \subseteq  \mathit{r}$,
we employ Ridge Regression \cite{rao1973linear} to learn the parameter $\phi_i$ for the regression over $\mathit{T}_i$, 
\begin{align}\label{equation-OLS-phi} 
\phi_i &= (\boldsymbol{\mathit{X}}^\top\boldsymbol{\mathit{X}} + \alpha\boldsymbol{\mathit{E}})^{-1}
\boldsymbol{\mathit{X}}^\top\boldsymbol{\mathit{Y}} 
\end{align}
where $\alpha$ is regularization parameter, $\boldsymbol{\mathit{E}}$ is identity matrix \cite{DBLP:books/daglib/0034861}, 
$\phi_i=\{\phi_i[C], \phi_i[A_1], \ldots, \allowbreak \phi_i[A_{m-1}]\}^\top$,
\begin{align}\label{equation-OLS-Y}
\boldsymbol{\mathit{Y}} = & 
\begin{pmatrix}
\mathit{t}_{1}[A_m] \\
\mathit{t}_{2}[A_m] \\
\vdots \\
\mathit{t}_{\ell}[A_m] \\
\end{pmatrix}
, \\ \label{equation-OLS-X}
\boldsymbol{\mathit{X}} = & 
\begin{pmatrix}
 1
 & \mathit{t}_{1}[A_1]		
 & \mathit{t}_{1}[A_2]
 & \ldots 
 & \mathit{t}_{1}[A_{m-1}] \\
 1
 & \mathit{t}_{2}[A_1]
 & \mathit{t}_{2}[A_2]
 & \ldots 
 & \mathit{t}_{2}[A_{m-1}] \\
\vdots & \vdots & \vdots & \ddots & \vdots \\
 1
 & \mathit{t}_{\ell}[A_1]
 & \mathit{t}_{\ell}[A_2]
 & \ldots 
 & \mathit{t}_{\ell}[A_{m-1}]
\end{pmatrix}. 
\end{align}

$\textsf{LR}(\mathcal{F}, \mathit{A}_\mathit{m}, \mathit{T}_i)$ 
in Line \ref{line-learning-parameter} 
computes the parameter $\phi_i$ 
over $\mathit{T}_i$.
It returns $\Phi$ the set of parameters $\phi_i$ for all tuples~$\mathit{t}_i$.

\begin{example}\label{example-profiling}
Consider relation $\mathit{r}$ in 
Figure \ref{fig-example-correct}.
Let $\ell = 4$.
According to Algorithm \ref{algorithm-learning}, 
we learn the individual regression for each tuple 
together with its neighbors
in $\mathit{r}$.
For $\mathit{t}_1$, we have $\mathit{T}_1 = \textsf{NN}(\mathit{t}_1, \{A_1\}, 4) = \{\mathit{t}_1, \mathit{t}_2, \mathit{t}_3, \mathit{t}_4\}$. 
The regression is learned from $\mathit{T}_1$ with parameter $\phi_1 = \{5.56, -0.87\}^\top$.
Similar computation applies to other tuples in $\mathit{r}$, having 
\begin{align*}
\textstyle
\Phi &=
\begin{pmatrix}
\phi_1 & \phi_2 & \ldots & \phi_8
\end{pmatrix}
=
\begin{pmatrix}
5.56 & 5.56 & \ldots & -4.36 \\
-0.87 & -0.87 & \ldots & 1.11 \\
\end{pmatrix} .
\end{align*}
\end{example}

\subsubsection{Handling Single Neighbor}
\label{sect-learning-single}

As mentioned in the introduction, a small number $\ell$ will lead to overfitting.
When $\ell=1$, 
the nearest neighbor returns only one tuple, i.e., $\mathit{T}_i=\{\mathit{t}_i\}$ which has the smallest distance 0 referring to Formula \ref{equation-distance}.
In this case, it is not sufficient to learn a proper regression model.
Hence, we directly set $\phi_i[C] = \mathit{t}_i[A_m]$ 
and $\phi_i[A_1] = \phi_i[A_2] = \ldots \phi_i[A_{m-1}] = 0$.

\subsubsection{Learning Complexity}
\label{sect-learning-complexity}

Line \ref{line-learning-lNN} in Algorithm \ref{algorithm-learning} 
takes $O(mn)$ time to compute distances of all tuples to $\mathit{t}_i$, 
and $O(\ell n)$ to find the $\ell$ nearest tuples
(advanced indexing and searching techniques could be applied, which is not the focus of this study).
Referring to Formula \ref{equation-OLS-phi}, 
Line \ref{line-learning-parameter} computes 
$\phi_i$ with cost $O(m^2\ell + m^3)$.
Thereby, the time complexity of Algorithm \ref{algorithm-learning} 
is $O(m n^2+\ell n^2 + m^2\ell n + m^3n)$.

\subsection{Imputation Phase}
\label{sect-imputation}

The \emph{imputation phase} utilizes the individual regression models of $\mathit{t}_\mathit{x}$'s neighbors from $\mathit{r}$
to compute the imputation candidates. 
Intuitively, in order to enhance the reliability, rather than only one neighbor, we consider the regressions of $\mathit{k}$ imputation neighbors 
(see Section \ref{experiment-imputation-number} for an evaluation on varying the number of imputation neighbors $\mathit{k}$).
These $\mathit{k}$ imputation candidates are then aggregated as the final imputation of $\mathit{t}_\mathit{x}$.

Algorithm \ref{algorithm-imputation} presents  major steps of the imputation phase: 

(S1) Imputation neighbors. 
Line \ref{line-kNN}
finds a set $\mathit{T}_\mathit{x}$ of $\mathit{k}$ nearest neighbors of incomplete tuple $\mathit{t}_\mathit{x}$
from relation $\mathit{r}$ on complete attributes $\mathcal{F}$, 
i.e., $\mathsf{NN}(\mathit{t}_\mathit{x}, \mathcal{F}, \mathit{k})$
as imputation neighbors.

(S2) Imputation candidates.
Line \ref{algorithm-local-candidate}
computes a possible imputation $\mathit{t}_\mathit{x}^{j}[\mathit{A}_\mathit{m}]$ by using the  regression of $\mathit{t}_\mathit{x}$'s neighbor $\mathit{t}_j$ with parameter $\phi_j$.

(S3) Combination. 
Line \ref{algorithm-imputation-combine} 
aggregates the candidates suggested by the regressions of all the $\mathit{t}_\mathit{x}$'s neighbors  in $\mathit{T}_\mathit{x}$
to form the final imputation $\mathit{t}_\mathit{x}'[\mathit{A}_\mathit{m}]$.

\begin{algorithm}[t]
\caption{Imputation($\mathit{t}_\mathit{x}, \mathit{k}, \Phi$)}
\label{algorithm-imputation}
 \KwIn{$\mathit{t}_\mathit{x}$ the tuple with missing value on attribute $\mathit{A}_\mathit{m}$, $\mathit{k}$ the number of imputation neighbors, 
 $\Phi$  individual regression parameters for all tuples in $\mathit{r}$} 
 \KwOut{imputation $\mathit{t}_\mathit{x}'[\mathit{A}_\mathit{m}]$}

 $\mathit{T}_\mathit{x} \leftarrow \mathsf{NN}(\mathit{t}_\mathit{x}, \mathcal{F}, \mathit{k})$\;
 \label{line-kNN} 
 \For{each $\mathit{t}_j \in \mathit{T}_\mathit{x}$}{
  $\mathit{t}_\mathit{x}^{j}[\mathit{A}_\mathit{m}] \leftarrow \mathsf{Candidate}(\phi_j, \mathit{t}_\mathit{x}[\mathcal{F}])$\;
  \label{algorithm-local-candidate}
 }
$\mathit{t}_\mathit{x}'[\mathit{A}_\mathit{m}] \leftarrow \mathsf{Combine}(\{\mathit{t}_\mathit{x}^{j}[\mathit{A}_\mathit{m}] \mid \mathit{t}_j\in\mathit{T}_\mathit{x}\})$\;
 \label{algorithm-imputation-combine}
\Return{$\mathit{t}_\mathit{x}'[\mathit{A}_\mathit{m}]$}
\end{algorithm}

\subsubsection{Find imputation neighbors for $\mathit{t}_\mathit{x}$ on complete attributes}

This step is the same as step (1) of $\textsf{kNN}$ imputation, i.e., 
find $k$ nearest neighbors $\mathit{T}_\mathit{x} = \textsf{NN}(\mathit{t}_\mathit{x}, \mathcal{F}, \mathit{k})$.
However, such neighbors are utilized in a different way. 
While the $\textsf{kNN}$ imputation directly aggregates the values on attribute $\mathit{A}_\mathit{m}$ of neighbors, e.g., in Formula \ref{equation-kNN}, 
our proposal considers the individual regression models w.r.t.\ these neighbors.

\subsubsection{Imputation via  individual regression of each neighbor}

For each neighbor $\mathit{t}_j \in \textsf{NN}(\mathit{t}_\mathit{x}, \mathcal{F}, \mathit{k})$,
we consider the individual regression with parameter $\phi_j$ learned in the learning phase by Formula \ref{equation-OLS-phi}.

Let $\mathit{t}_\mathit{x}^j$ denote the imputation candidate suggested by the regression of the neighbor $\mathit{t}_j$. 
Referring to Formula \ref{equation-regression}, 
we have 
\begin{align}
\mathit{t}_\mathit{x}^j[\mathit{A}_\mathit{m}] &= (1,\mathit{t}_\mathit{x}[\mathcal{F}]) \phi_j + \varepsilon_j, 
\end{align} 
where $\varepsilon_j$ is the error term of the regression w.r.t.\ $\mathit{t}_j$. 
It is common to omit the error term $\varepsilon_j$ \cite{rubin2004multiple} and thus the imputation candidate of the neighbor $\mathit{t}_j$ is 
\begin{align}
\mathit{t}_\mathit{x}^j[\mathit{A}_\mathit{m}] &= (1,\mathit{t}_\mathit{x}[\mathcal{F}]) \phi_j
\end{align}

\begin{figure}[t]
\centering
\includegraphics[width=0.7\figwidths]{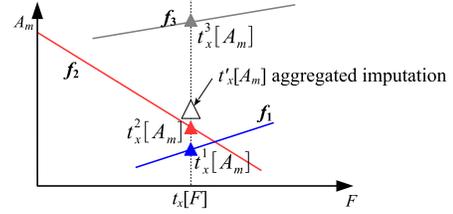}
\caption{Intuition example of combining imputation candidates}
\label{fig:example-weight}
\end{figure}

\subsubsection{Aggregating individual imputation candidates}
\label{sect-combine-imputation}

In the imputation phase, 
the tuple $\mathit{t}_\mathit{x}$ with missing values finds complete tuples $\mathit{t}_i$ as its neighbors,  
and proposes to utilize the aforesaid individually learned models of $\mathit{t}_i$. 
Again, owing to heterogeneity (the argument to learn individualized models), 
not all the neighbors $\mathit{t}_i$ may share the same models with $\mathit{t}_\mathit{x}$, 
i.e., 
not all the neighbors $\mathit{t}_i$ would provide a model leading to the correct value for imputing $\mathit{t}_\mathit{x}$. 
Arbitrarily selecting one $\mathit{t}_i$ may lead to the wrong imputation.
A neighbor $\mathit{t}_i$ with closer distance to $\mathit{t}_\mathit{x}$ on the complete attribute $\mathcal{F}$
does not denote that its model applies to $\mathit{t}_\mathit{x}$ either. 
Thereby, we propose a weighted aggregation of the imputation candidates provided by the models of different neighbors $\mathit{t}_i$, 
where the candidate values vote for each other. 

The aggregated imputation result is defined by
\begin{align}\label{equation-imputation}
\mathit{t}_\mathit{x}'[\mathit{A}_\mathit{m}] = 
\sum\limits_{\mathit{t}_j\in\mathit{T}_\mathit{x}}
\mathit{t}_\mathit{x}^j[\mathit{A}_\mathit{m}]\cdot \mathit{w}_{\mathit{x}j}, 
\end{align}
where 
$\mathit{t}_\mathit{x}^j[\mathit{A}_\mathit{m}]$ is the imputation candidate suggested by the imputation neighbor $\mathit{t}_j \in \textsf{NN}(\mathit{t}_\mathit{x},\mathcal{F}, \mathit{k})$, and
$\mathit{w}_{\mathit{x}j}$ is the weight of candidate 
$\mathit{t}_\mathit{x}^j[\mathit{A}_\mathit{m}]$ in aggregation.

Intuitively, 
we propose to let the candidate values $\mathit{t}_\mathit{x}^i[\mathit{A}_m]$ 
(provided by the models from different neighbor tuples $\mathit{t}_i$) vote for each other, 
via a weighted aggregation function. 
Similar to the idea of majority voting, 
those candidate values close with each other are more likely to be the imputation and may assign higher weights in aggregation, while outliers could be largely ignored with lower aggregation weights.
For instance, 
in Figure \ref{fig:example-weight}, 
the candidates $\mathit{t}_\mathit{x}^1[\mathit{A}_\mathit{m}]$ and $\mathit{t}_\mathit{x}^2[\mathit{A}_\mathit{m}]$ suggested by models $\mathit{f}_1$ and $\mathit{f}_2$, respectively, are close and agree with each other. 
In contrast, the other candidate $\mathit{t}_\mathit{x}^3[\mathit{A}_\mathit{m}]$ by $\mathit{f}_3$ is outlying (due to heterogeneity), 
and would be largely ignored with lower aggregation weights.

In this sense, 
we consider the distances of a candidate 
$\mathit{t}_\mathit{x}^i[\mathit{A}_\mathit{m}]$ to the other candidates, 
\begin{align}\label{equation-distance}
\mathit{c}_{\mathit{x}i}= \sum\limits_{j=1}^{k}
\left| \mathit{t}_\mathit{x}^i[\mathit{A}_\mathit{m}] - \mathit{t}_\mathit{x}^j[\mathit{A}_\mathit{m}] \right|.
\end{align}
Following the intuition that candidates close to other (i.e., having smaller $\mathit{c}_{\mathit{x}i}$) should assign larger weight, 
we define 
\begin{align}\label{equation-weight}
\textstyle
\mathit{w}_{\mathit{x}i} &= \frac{
\mathit{c}_{\mathit{x}i}^{-1}
}{
\sum\limits_{j=1}^{k}\mathit{c}_{\mathit{x}j}^{-1}
},
\end{align}
having $\sum\limits_{j=1}^{k}
\mathit{w}_{\mathit{x}j} = 1$.

\begin{example}
\label{example-imputation}
Let $\mathit{k} = 3, \ell = 4$. 
The imputation starts from the  parameter $\Phi$ 
learned in Example \ref{example-profiling}.
Algorithm \ref{algorithm-imputation} performs in three steps: 
(1) Find imputation neighbors for the incomplete tuple $\mathit{t}_\mathit{x}$, having 
$\mathit{T}_\mathit{x} = \textsf{NN}(\mathit{t}_\mathit{x}, \{A_1\}, 3) = \{\mathit{t}_5, \mathit{t}_4, \mathit{t}_6\}$
(2) Compute the imputation candidate via the individual regression of each neighbor.
For $\mathit{t}_5$, referring to the  regression model $\textsf{LR}(\{A_1\}, A_2, \mathit{T}_5)$
with parameter $\phi_5 = (-4.36, 1.11)^\top$, 
the imputation candidate is computed by $\mathit{t}_\mathit{x}^5[A_2] = (1, 5)(-4.36, 1.11)^\top = 1.19$.
Similar computation applies to neighbors $\mathit{t}_4$ and $\mathit{t}_6$, having
$\mathit{t}_\mathit{x}^4[A_2] = (1, 5)(5.56, -0.87)^\top = 1.21, \allowbreak
\mathit{t}_\mathit{x}^6[A_2] = (1, 5)(-4.36, 1.11)^\top \allowbreak = 1.19$.
(3) Aggregating the aforesaid  imputation candidates. 
Following Formula \ref{equation-distance}, we can compute the distance for each imputation candidates as $\mathit{c}_{\mathit{x}5} = \mathit{c}_{\mathit{x}6} = 0.02, \mathit{c}_{\mathit{x}4} = 0.04$.
Thus the aggregated imputation by Formula \ref{equation-imputation} is 
$
\mathit{t}_\mathit{x}'[A_2] = 1.19 * \frac{50}{125} + 1.21 * \frac{25}{125} + 1.19 * \frac{50}{125} = 1.194.
$
\end{example}

\subsubsection{Imputation Complexity}
\label{sect-imputation-complexity}

Similar to the analysis in Section \ref{sect-learning-complexity}, 
Line \ref{line-kNN} in Algorithm \ref{algorithm-imputation} searches the $\mathit{k}$ nearest neighbors with cost $O(mn+kn)$.
The imputation candidates w.r.t.\ $\mathit{k}$ imputation neighbors are then computed and combined in Lines \ref{algorithm-local-candidate} and \ref{algorithm-imputation-combine} with cost $O(\mathit{m}\mathit{k}+\mathit{k}^2)$. 
Thereby, the time complexity of Algorithm \ref{algorithm-imputation} is 
$O(mn + kn)$.

\subsection{Discussion on Overheads and Benefits}
\label{sect-discussion-framework}

Learning over individual tuples and their $\ell$ neighbors is a bit more expensive
than learning a global model over all the $\mathit{n}$ tuples. 
Referring to \cite{rao1973linear}, 
the cost of learning a regression model over $\mathit{n}$ tuples is $O(\mathit{m}^2\mathit{n} + \mathit{m}^3)$, 
while 
the cost of learning $\mathit{n}$ individual models for $\mathit{n}$ tuples given their $\ell$ neighbors is $O((\mathit{m}^2\ell + \mathit{m}^3)\mathit{n})$.
Nevertheless, both complexities are linear w.r.t.\ the number of tuples $\mathit{n}$.
In particular, all these models (global or individual) could be offline learned over complete tuples, 
and directly used in online imputing the missing values of various incomplete tuples.

The benefit of the imputation by individual models (\textsf{IIM}) would be the clearly higher accuracy than that of \textsf{GLR} with a single (inaccurate) global model, 
as shown in Table~\ref{table-all-type}. 

\section{Subsuming Existing Methods}
\label{sect-subsume}

To illustrate the rationale of the proposed \textsf{IIM} imputation, 
in this section, we theoretically prove that some existing methods (\textsf{kNN} \cite{altman1992introduction} and \textsf{GLR} \cite{little2014statistical} introduced in Sections \ref{sect-related-tuple} and \ref{sect-related-attribute}) are indeed special cases of our \textsf{IIM} under extreme settings (i.e., $\ell=1$ or $\ell=\mathit{n}$).
It further motivates us to adaptively determine a distinct number of learning neighbors $\ell$ (in between the extreme $1$ and $\mathit{n}$) for each tuple in Section \ref{sect-adaptive}.

\subsection{Subsuming kNN}

First, we show that \textsf{IIM} subsumes \textsf{kNN} 
by considering only one learning neighbor in individual learning, i.e., $\ell= 1$.

\begin{proposition}[Subsume \textsf{kNN}]\label{the:l0-equal}
When we consider 
a fixed number of learning neighbors $\ell= 1$
and a uniform weight of imputation candidate 
$w_{\mathit{x}i} = \frac{
1
}{
|\mathit{T}_{\mathit{x}}|
}$, 
the proposed \textsf{IIM} algorithm is equivalent to the \textsf{kNN} imputation.
\end{proposition}
\begin{proof}
When $\ell = 1$, 
for any $\mathit{t}_i\in\mathit{r}$, 
its neighbor is $\mathit{T}_i = \textsf{NN}(\mathit{t}_i, \mathcal{F}, \ell) = \{\mathit{t}_i\}$, 
i.e., itself for learning the individual model.
It leads to the case of single neighbor in learning, described in Section \ref{sect-learning-single}.
That is, we have
$
\phi_i[C] = \mathit{t}_i[A_m] ,
\phi_i[A_j] = 0, 1 \leq j \leq m-1
$.

Consider the $\mathit{k}$ imputation neighbors of incomplete tuple $\mathit{t}_\mathit{x}$, i.e.,  
$\mathit{T}_\mathit{x} 
= \textsf{NN}(\mathit{t}_\mathit{x}, \mathcal{F}, \mathit{k}) \subseteq \mathit{r}$.
For each $\mathit{t}_i \in \mathit{T}_\mathit{x}$, 
referring to the individual regression with parameter $\phi_i$, 
we have $\mathit{t}_\mathit{x}^i[A_m] = \mathit{t}_i[A_m]$.
Referring to the uniform weight of imputation candidates
$w_{\mathit{x}i} = \frac{
1
}{
|\mathit{T}_{\mathit{x}}|
}$,
the imputation obtained by Formula \ref{equation-imputation} thus has 
$
\mathit{t}_\mathit{x}'[A_m] = \frac{
\sum\limits_{\mathit{t}_i\in\mathit{T}_\mathit{x}}\mathit{t}_i[A_m]
}{
|\mathit{T}_\mathit{x}|
},
$
which is exactly the same as the \textsf{kNN} imputation in Formula \ref{equation-kNN}.
\end{proof}

\subsection{Subsuming GLR}

Moreover, we prove that \textsf{IIM} subsumes \textsf{GLR} 
by considering all the tuples in $\mathit{r}$ as the learning neighbors in individual learning, i.e., $\ell= \mathit{n} = |\mathit{r}|$.

\begin{proposition}[Subsume \textsf{GLR}]\label{the-LR-equal}
When we consider a fixed
number of learning neighbors $\ell= \mathit{n} = |\mathit{r}|$,
the \textsf{IIM} algorithm is equivalent to the \textsf{GLR} imputation.
\end{proposition}
\begin{proof}
When $\ell = \mathit{n}$, 
for any $\mathit{t}_i\in\mathit{r}$, 
its learning neighbors are $\mathit{T}_i = \textsf{NN}(\mathit{t}_i, \mathcal{F}, \ell) = \mathit{r}$, 
i.e., all the complete tuples.
Let 
$\phi_\mathit{r}$ be the parameter of the global regression learned from the entire $\mathit{r}$.
We have 
$
\phi_i = \phi_\mathit{r}.
$

Consider the $\mathit{k}$ imputation neighbors of incomplete tuple $\mathit{t}_\mathit{x}$, i.e.,  
$\mathit{T}_\mathit{x} 
= \textsf{NN}(\mathit{t}_\mathit{x}, \mathcal{F}, \mathit{k}) \subseteq \mathit{r}$.
For each $\mathit{t}_i \in \mathit{T}_\mathit{x}$, 
referring to the individual regression with parameter $\phi_i=\phi_\mathit{r}$, 
we have $\mathit{t}_\mathit{x}^i[A_m] = (1, \mathit{t}_\mathit{x}[\mathcal{F}])\phi_\mathit{r}$.
The imputation obtained by Formula \ref{equation-imputation} thus has 
$
\mathit{t}_\mathit{x}'[A_m] = 
(1, \mathit{t}_\mathit{x}[\mathcal{F}])\phi_\mathit{r},
$
which is exactly the same as the \textsf{GLR} imputation in Formula \ref{equation-imputation-global}.
\end{proof}

\section{Adaptive Learning}
\label{sect-adaptive}

In the learning phase in Section \ref{sect-learning}, 
a fixed number $\ell$ of learning neighbors is considered for all the tuples in $\mathit{r}$ in Algorithm \ref{algorithm-learning}. 
There are two issues to concern:
(1) how to determine a proper number $\ell$ of neighbors for learning; and 
(2) different tuples may prefer a distinct number $\ell$ of learning neighbors,  
owing to heterogeneity. 

In Section \ref{sect-adaptive-learning-algorithm}, 
we consider the various candidate regression models learned under different $\ell$ for a tuple. 
The adaptive learning (Algorithm \ref{algorithm-adaptive})
selects a proper $\ell$ as well as the corresponding model for each tuple.
Intuitively, to evaluate whether a model learned under some $\ell$ is proper, 
we may consider a set of complete tuples as validation data, 
and see which learned models can best impute the validation tuples (truth is known in the complete validation tuple).

In Section \ref{sect-incremental-computation},
to 
efficiently
learn the candidate regression models under various $\ell$ for a tuple, 
we devise an incremental computing scheme. 
Remarkably, it reduces the time complexity of individual learning from linear to constant.

\subsection{Adaptive Learning with Validation}
\label{sect-adaptive-learning-algorithm}

Algorithm \ref{algorithm-adaptive} presents the procedure of adaptively learning a proper regression model from $\mathcal{F}$ to $\mathit{A}_\mathit{m}$ for each complete tuple $\mathit{t}_i\in\mathit{r}$ under various number $\ell$ of learning neighbors. 

First, Line \ref{line-call-learn} learns the candidate models under various $\ell$ for all tuples in $\mathit{r}$, denoted by $\Phi^{(\ell)}$, 
by call the Learning Algorithm \ref{algorithm-learning}.
(Advanced incremental computation is devised among different $\ell$ in Section \ref{sect-incremental-computation}.)

We consider the complete tuples in $\mathit{r}$ as the validation set. 
For each $\mathit{t}_j\in\mathit{r}$ employed as a validation tuple, 
we assume its $\mathit{t}_j[\mathit{A}_\mathit{m}]$ is missing. 
The original complete value $\mathit{v}$ of $\mathit{t}_j[\mathit{A}_\mathit{m}]$ is directly used to evaluate 
how the models from $\mathit{t}_i$ (neighbor of $\mathit{t}_j$)  
could accurately impute~$\mathit{t}_j[\mathit{A}_\mathit{m}]$.

It is worth noting that 
the model of tuple $\mathit{t}_i$ learned over a number of  $\ell$ learning neighbors
can be applied multiple times to impute various $\mathit{t}_j$.
The $cost[i][\ell]$ in Line \ref{line-adaptive-cost} in Algorithm~\ref{algorithm-adaptive}
denotes the total difference between the truths and the imputations given different validation tuples $\mathit{t}_j$.
A model with smaller $cost[i][\ell]$ means more accurate imputation when applied,
and thus is preferred in Line \ref{line-min-imputation-cost}.
This extra overhead is necessary, since we want to select a proper $\ell$ that performs well in general for imputing potentially all the nearby tuples $\mathit{t}_j$.

\begin{algorithm}[t]
\caption{Adaptive($\mathit{r}, \mathcal{F}$, $\mathit{A}_\mathit{m}$)}
\label{algorithm-adaptive}
 \KwIn{relation $\mathit{r}$ of complete tuples, complete attributes $\mathcal{F}$, incomplete attribute $\mathit{A}_\mathit{m}$}
 \KwOut{$\Phi$ the set of regression parameters $\phi_i$ learned for all tuples $\mathit{t}_i$ in $\mathit{r}$}
 \For{$\ell \leftarrow 1$ \KwTo $n$}{ \label{line-adaptive-ell}
   $\Phi^{(\ell)} \leftarrow \textrm{Learning}(\mathit{r}, \ell, \mathcal{F}, \mathit{A}_\mathit{m})$\; \label{line-call-learn}
 }
 \For{each $\mathit{t}_j \in \mathit{r}$}{ \label{line-adaptive-tuple}
   $\mathit{T}_j \leftarrow \textsf{NN}(\mathit{t}_j, \mathcal{F}, \mathit{k})$\; \label{line-adaptive-lNN} 
   \For{each $\mathit{t}_i \in \mathit{T}_j $}{
     \For{$\ell \leftarrow 1$ \KwTo $n$}{ \label{line-all-ell}
     $cost[i][\ell] += 
     \left(
       \mathit{t}_{j}[\mathit{A}_\mathit{m}]-
       (1, \mathit{t}_{j}[\mathcal{F}])\phi_i^{(\ell)}
       \right)^2$\; \label{line-adaptive-cost}
     } 
   } 
 } 
 \For{$i \leftarrow 1$ \KwTo $n$}{
  $\ell_i^* \leftarrow \argmin_{\ell\in[1,\mathit{n}]}{cost[i][\ell]}$\; \label{line-min-imputation-cost}
  $\phi_i \leftarrow \phi_i^{(\ell_i^*)}$\;
 }
\Return{$\Phi$ }
\end{algorithm}

\begin{example}
\label{example-determination-straightforward}
Consider relation $\mathit{r}$ in Figure \ref{fig-example-correct}.
Suppose that we have learned candidate models under various $\ell$ for all the tuples in Line \ref{line-call-learn} 
in Algorithm \ref{algorithm-adaptive}.
Given $\mathit{k} = 3$, 
we determine a proper model for each tuple
from the candidate models $\Phi^{(\ell)}$.

Let $\mathit{t}_1$ be the validation tuple.  
Line \ref{line-adaptive-lNN} finds \textsf{kNN} of $\mathit{t}_1$ on the complete attribute $\mathit{A}_1$, i.e., $\mathit{T}_1 = \{\mathit{t}_2, \mathit{t}_3, \mathit{t}_4\}$.
For each tuple in $\mathit{T}_1$, say $\mathit{t}_2$, 
the difference between imputation by various candidate models of $\mathit{t}_2$ and the truth of $\mathit{t}_1[\mathit{A}_2]$ are recorded, 
\begin{align*}
\textstyle
cost[2][1] &= (5.8 - (1, 0)(4.35, 0)^\top)^2 = 2.1, \\
cost[2][2] &= (5.8 - (1, 0)(5.79, -1.49)^\top)^2 = 0.0001, \\
 \dots \\
cost[2][8] &= (5.8 - (1, 0)(4.41, -0.01)^\top)^2 = 1.93.
\end{align*}
Line \ref{line-adaptive-cost} aggregates such difference costs on all the tuples in $\mathit{r}$ (as validation set) in addition to the aforesaid $\mathit{t}_1$.
We have 
$\{cost[2][1], cost[2][2], \dots, cost[2][8]\} = $
\{3.73, 3.67, 0.31, \textbf{0.09}, 1.47, 2.36, 3.03, 3.65\}.
Finally, $\ell_2^* = 4$ with the minimal $cost[2][4]$ is selected and $\phi_2 = \phi_2^{(4)} = \{5.56, -0.87\}^\top$ is returned
as the parameter of the model for $\mathit{t}_2$.
\end{example}

\subsubsection{Adaptive Learning Complexity}
\label{sect-determination-complexity}

We can precompute once the nearest neighbors for all tuples in $\mathit{r}$ with cost $O((m+n)n^2) = O(n^3)$ and directly use them in learning individual model for a certain $\ell$.
According to Algorithm \ref{algorithm-learning}, 
the learning phase computes $\phi_i$ with cost $O(m^2\ell + m^3)$ for a certain $\ell$ and cost  
$O(m^2n^2)$ for all possible $\ell$ from $1$ to $n$.
For each tuple $\mathit{t}_i$, 
the cost for computing difference is $O(kn)$.
Thus the time cost from Line \ref{line-adaptive-tuple} to Line \ref{line-adaptive-cost} is $O(kn^2)$.
Obviously, it costs $O(n^2)$ to find the proper $\ell^*$ for all the tuples.
Finally, the time complexity of Algorithm \ref{algorithm-adaptive} is 
$O(m^2n^2 + n^3)$.

\subsubsection{Approximation via Stepping}
\label{sect-stepping}

When considering various $\ell$ in Line \ref{line-adaptive-ell} in Algorithm \ref{algorithm-adaptive}, 
instead of increasing 1 in each iteration, i.e., $\ell=\ell+1$, 
we may increase more, say $\ell=\ell+\mathit{h}, \mathit{h}\geq1$ in stepping.
The time cost by stepping significantly reduces, 
from $O(\mathit{m}^2\mathit{n}^3)$ to $O(\mathit{m}^2\mathit{n}^3/\mathit{h})$.
However, it may miss a better model in between $\ell$ and $\ell+\mathit{h}$. 
Therefore, stepping is a tradeoff between efficiency and accuracy.
(See Section \ref{sect-experiment-stepping} for results under various  stepping $\mathit{h}$.)

\begin{example}
For stepping $\mathit{h} = 3$, only the $\ell$ values $\{1, 4, 7\}$ will be considered, instead of all 8 possible $\ell$.
Similar to Example \ref{example-determination-straightforward}, 
for tuple $\mathit{t}_2$, it computes $cost[2][1] = 3.73, cost[2][4]=0.09, cost[2][7]=3.03$.
Finally, $\ell_2^*=4$ is selected and $\Phi_2 = \{5.56, -0.87\}^\top$ is returned. 
\end{example}

\subsection{Incremental Computation}
\label{sect-incremental-computation}

For a specific $\ell$, Line \ref{line-call-learn} in Algorithm \ref{algorithm-adaptive} 
calls the individual Learning Algorithm \ref{algorithm-learning} 
starting from scratch, 
without utilizing any results from the previous learning,
e.g., $\ell-1$.
It is worth noting that the $\ell-1$ learning neighbors of a tuple are always subsumed in the corresponding $\ell$ neighbors (Formula \ref{equation-neighbor-subsume}). 
Intuitively, the learning computation on $\ell-1$ neighbors has no need to repeat in the learning over $\ell$ neighbors.

\subsubsection{Incremental Learning}

Let 
$\mathit{T}_i^{(\ell)}=\textsf{NN}(\mathit{t}_i, \mathcal{F}, \ell)
=\{\mathit{t}_1,\dots,\mathit{t}_\ell\}$ 
denote the set of $\ell$ nearest neighbors of $\mathit{t}_i\in\mathit{r}$,
and 
$\phi_i^{(\ell)}$ be the parameter of the individual regression learned from $\mathit{T}_i^{(\ell)}$
by Formula \ref{equation-OLS-phi}.
As aforesaid, 
subsumption relationship exists among the sets of nearest neighbors with different sizes $\ell$. 
That is, for any tuple $\mathit{t}_i\in\mathit{r}$, $\mathit{h}\geq1$, we have 
\begin{align}\label{equation-neighbor-subsume}
\mathit{T}_i^{(\ell)}=\textsf{NN}(\mathit{t}_i, \mathcal{F}, \ell) \subset
\mathit{T}_i^{(\ell+\mathit{h})} = \textsf{NN}(\mathit{t}_i, \mathcal{F}, \ell+\mathit{h}).
\end{align}
Intuitively, 
the regression model, e.g., $\phi_i^{(\ell+\mathit{h})}$ learned over $\mathit{T}_i^{(\ell+\mathit{h})}$, 
can be incrementally computed from the previous results, i.e., $\phi_i^{(\ell)}$ learned over $\mathit{T}_i^{(\ell)}$, 
in Proposition \ref{the-incremental-bottom-up},
rather than starting from scratch
in Algorithm \ref{algorithm-learning}.
Remarkably, we show in Table \ref{table-increment-time} that 
the incremental computation reduces the learning complexity from linear to constant (in terms of the number $\ell$).

Let 
$\mathit{T}^{(\ell+\mathit{h})}_i = \mathsf{NN}(\mathit{t}_i, \mathcal{F}, \ell+\mathit{h})=\{\mathit{t}_1,\dots,\mathit{t}_\ell,  \mathit{t}_{\ell+1}, \dots, \allowbreak \mathit{t}_{\ell+\mathit{h}}\}$, 
having 
\begin{align}\label{equation-neighbor-increment}
\mathit{T}_i^{(\ell+\mathit{h})}\setminus\mathit{T}_i^{(\ell)}=\{\mathit{t}_{\ell+1},\dots,\mathit{t}_{\ell+\mathit{h}}\}.
\end{align}

To represent the increment, 
we rewrite $\boldsymbol{\mathit{Y}}^{(\ell+\mathit{h})}$ in Formula \ref{equation-OLS-Y} 
and $\boldsymbol{\mathit{X}}^{(\ell+\mathit{h})}$ in Formula \ref{equation-OLS-X}
as follows,
\begin{align}\label{equation-OLS-incerment-Y} 
\boldsymbol{\mathit{Y}}^{(\ell+\mathit{h})} = & 
\begin{pmatrix}
\boldsymbol{\mathit{Y}}^{(\ell)} \\
\mathit{t}_{\ell+1}[A_m] \\
\vdots \\
\mathit{t}_{\ell+\mathit{h}}[A_m] \\
\end{pmatrix}
=
\begin{pmatrix}
\boldsymbol{\mathit{Y}}^{(\ell)} \\
\boldsymbol{\mathit{Y}}^{(\ell, \Delta_{\mathit{h}})} \\
\end{pmatrix}, \\
\label{equation-OLS-incerment-X} 
\boldsymbol{\mathit{X}}^{(\ell+\mathit{h})} = & 
\begin{pmatrix}
& &
\boldsymbol{\mathit{X}}^{(\ell)} \\
 1
 & \mathit{t}_{\ell+1}[A_1]
 & \ldots 
 & \mathit{t}_{\ell+1}[A_{m-1}] \\
 1
 & \mathit{t}_{\ell+2}[A_1]
 & \ldots 
 & \mathit{t}_{\ell+2}[A_{m-1}] \\
\vdots & \vdots & \ddots & \vdots \\
 1
 & \mathit{t}_{\ell+\mathit{h}}[A_1]
 & \ldots 
 & \mathit{t}_{\ell+\mathit{h}}[A_{m-1}]
\end{pmatrix}
=
\begin{pmatrix}
\boldsymbol{\mathit{X}}^{(\ell)} \\
\boldsymbol{\mathit{X}}^{(\ell, \Delta_{\mathit{h}})}
\end{pmatrix},
\end{align}
where 
$\boldsymbol{\mathit{Y}}^{(\ell+\mathit{h})}$ is an $(\mathit{\ell+\mathit{h}})\times1$ matrix, 
and 
$\boldsymbol{\mathit{X}}^{(\ell+\mathit{h})}$ is an $(\mathit{\ell+\mathit{h}})\times\mathit{m}$ matrix.

To incrementally compute $\phi_i^{(\ell+\mathit{h})}$ by Formula \ref{equation-OLS-phi}, 
we  define
\begin{align}
\label{equation-mid-matrix-U}
\boldsymbol{\mathit{U}}^{(\ell+\mathit{h})} &= 
(\boldsymbol{\mathit{X}}^{(\ell+\mathit{h})})^\top
\boldsymbol{\mathit{X}}^{(\ell+\mathit{h})} ,
\\ 
\label{equation-mid-matrix-V}
\boldsymbol{\mathit{V}}^{(\ell+\mathit{h})} &= 
(\boldsymbol{\mathit{X}}^{(\ell+\mathit{h})})^\top
\boldsymbol{\mathit{Y}}^{(\ell+\mathit{h})} ,
\end{align}
where 
$\boldsymbol{\mathit{U}}^{(\ell+\mathit{h})}$ is an $\mathit{m}\times\mathit{m}$ matrix, 
and 
$\boldsymbol{\mathit{V}}^{(\ell+\mathit{h})}$ is an $\mathit{m}\times1$ matrix
with both sizes independent of $\ell$ and $\mathit{h}$.

Formula \ref{equation-OLS-phi} for learning the parameter can be rewritten by
\begin{align}
\label{equation-mid-matrix-phi}
\phi_i^{(\ell+\mathit{h})} &= (\boldsymbol{\mathit{U}}^{(\ell+\mathit{h})} + \alpha\boldsymbol{\mathit{E}})^{-1}
\boldsymbol{\mathit{V}}^{(\ell+\mathit{h})}.
\end{align}

We show in the proposition below that 
$\boldsymbol{\mathit{U}}^{(\ell+\mathit{h})}$ and $\boldsymbol{\mathit{V}}^{(\ell+\mathit{h})}$ 
can be incrementally computed from 
$\boldsymbol{\mathit{U}}^{(\ell)}$ and $\boldsymbol{\mathit{V}}^{(\ell)}$,
together with 
$\boldsymbol{\mathit{Y}}^{(\ell, \Delta_{\mathit{h}})}$
and 
$\boldsymbol{\mathit{X}}^{(\ell, \Delta_{\mathit{h}})}$
defined in Formulas \ref{equation-OLS-incerment-Y} and \ref{equation-OLS-incerment-X}.

\begin{proposition}\label{the-incremental-bottom-up}
$\boldsymbol{\mathit{U}}^{(\ell+\mathit{h})}, \boldsymbol{\mathit{V}}^{(\ell+\mathit{h})}$ 
could be incrementally computed from 
$\boldsymbol{\mathit{U}}^{(\ell)}, \boldsymbol{\mathit{V}}^{(\ell)}$, 
having
\begin{align}
\label{equation-estimate-recursive-u}
\boldsymbol{\mathit{U}}^{(\ell+\mathit{h})} &= 
\boldsymbol{\mathit{U}}^{(\ell)} + 
(\boldsymbol{\mathit{X}}^{(\ell, \Delta_{\mathit{h}})})^\top
\boldsymbol{\mathit{X}}^{(\ell, \Delta_{\mathit{h}})} \\
\label{equation-estimate-recursive-v}
\boldsymbol{\mathit{V}}^{(\ell+\mathit{h})} &= 
\boldsymbol{\mathit{V}}^{(\ell)} + 
(\boldsymbol{\mathit{X}}^{(\ell, \Delta_{\mathit{h}})})^\top
\boldsymbol{\mathit{Y}}^{(\ell, \Delta_{\mathit{h}})}
\end{align}
where $\ell\in[1,n)$ and $h \in [1, n-\ell]$. 
\end{proposition}
\begin{proof}
We show the correctness of Formulas \ref{equation-estimate-recursive-u} and \ref{equation-estimate-recursive-v}, respectively, as follows.

(1) For $\boldsymbol{\mathit{U}}^{(\ell+\mathit{h})}$, we have
\begin{align*}
\boldsymbol{\mathit{U}}^{(\ell+\mathit{h})} &= 
(\boldsymbol{\mathit{X}}^{(\ell+\mathit{h})})^\top
\boldsymbol{\mathit{X}}^{(\ell+\mathit{h})} \\
&= 
\begin{pmatrix}
(\boldsymbol{\mathit{X}}^{(\ell)})^\top & 
(\boldsymbol{\mathit{X}}^{(\ell, \Delta_{\mathit{h}})})^\top
\end{pmatrix}
\begin{pmatrix}
\boldsymbol{\mathit{X}}^{(\ell)} \\
\boldsymbol{\mathit{X}}^{(\ell, \Delta_{\mathit{h}})}
\end{pmatrix}
\\
&=
(\boldsymbol{\mathit{X}}^{(\ell)})^\top
\boldsymbol{\mathit{X}}^{(\ell)} +
(\boldsymbol{\mathit{X}}^{(\ell, \Delta_{\mathit{h}})})^\top
\boldsymbol{\mathit{X}}^{(\ell, \Delta_{\mathit{h}})} \\
&=
\boldsymbol{\mathit{U}}^{(\ell)} + 
(\boldsymbol{\mathit{X}}^{(\ell, \Delta_{\mathit{h}})})^\top
\boldsymbol{\mathit{X}}^{(\ell, \Delta_{\mathit{h}})}.
\end{align*}

(2) For $\boldsymbol{\mathit{V}}^{(\ell+\mathit{h})}$, we have
\begin{align*}
\boldsymbol{\mathit{V}}^{(\ell+\mathit{h})} &= 
(\boldsymbol{\mathit{X}}^{(\ell+\mathit{h})})^\top
\boldsymbol{\mathit{Y}}^{(\ell+\mathit{h})} \\
&= 
\begin{pmatrix}
(\boldsymbol{\mathit{X}}^{(\ell)})^\top & 
(\boldsymbol{\mathit{X}}^{(\ell, \Delta_{\mathit{h}})})^\top
\end{pmatrix}
\begin{pmatrix}
\boldsymbol{\mathit{Y}}^{(\ell)} \\
\boldsymbol{\mathit{Y}}^{(\ell, \Delta_{\mathit{h}})}
\end{pmatrix}
\\
&=
(\boldsymbol{\mathit{X}}^{(\ell)})^\top
\boldsymbol{\mathit{Y}}^{(\ell)} +
(\boldsymbol{\mathit{X}}^{(\ell, \Delta_{\mathit{h}})})^\top
\boldsymbol{\mathit{Y}}^{(\ell, \Delta_{\mathit{h}})} \\
&=
\boldsymbol{\mathit{V}}^{(\ell)} + 
(\boldsymbol{\mathit{X}}^{(\ell, \Delta_{\mathit{h}})})^\top
\boldsymbol{\mathit{Y}}^{(\ell, \Delta_{\mathit{h}})}.
\end{align*}
\end{proof}

\begin{example}\label{example-bottom-up}
Suppose that learning on $\mathit{t}_1$ with $\ell=3$ has been performed, having
$\textsf{NN}(\mathit{t}_1, \{A_1\}, 3) = \{\mathit{t}_1, \mathit{t}_2, \mathit{t}_3\}$,
\begin{align*}
\textstyle
\boldsymbol{\mathit{U}}^{(3)} &= 
(\boldsymbol{\mathit{X}}^{(3)})^\top
\boldsymbol{\mathit{X}}^{(3)}
= 
\begin{pmatrix}
1 & 1 & 1 \\
0 & 0.8 & 1.9
\end{pmatrix}
\begin{pmatrix}
1 & 0  \\
1 & 0.8  \\
1 & 1.9
\end{pmatrix}
, \\
\boldsymbol{\mathit{V}}^{(3)} &= 
(\boldsymbol{\mathit{X}}^{(3)})^\top
\boldsymbol{\mathit{Y}}^{(3)}
= 
\begin{pmatrix}
1 & 1 & 1 \\
0 & 0.8 & 1.9
\end{pmatrix}
\begin{pmatrix}
5.8 \\
4.6 \\
3.8
\end{pmatrix}
, \\
\phi_1^{(3)} &= 
(\boldsymbol{\mathit{U}}^{(3)} + \alpha\boldsymbol{\mathit{E}})^{-1}
\boldsymbol{\mathit{V}}^{(3)}
=
\begin{pmatrix}
5.66 \\
-1.03
\end{pmatrix}.
\end{align*}

Now we want to learn the parameter $\phi_1^{(4)}$ of $\mathit{t}_1$ for $\ell = 4$, 
having $\textsf{NN}(\mathit{t}_1, \{A_1\}, 4) = \textsf{NN}(\mathit{t}_1, \{A_1\}, 3) \cup \{\mathit{t}_4\}$.
Instead of recomputing entirely the matrices $\boldsymbol{\mathit{U}}^{(4)}, \boldsymbol{\mathit{V}}^{(4)}$, 
they can be incrementally computed from $\boldsymbol{\mathit{U}}^{(3)},
\boldsymbol{\mathit{V}}^{(3)}$.
Specifically, given $\boldsymbol{\mathit{X}}^{(3, 1)} = 
\begin{pmatrix}
1 & 2.9
\end{pmatrix}$ and 
$\boldsymbol{\mathit{Y}}^{(3, 1)} = 
\begin{pmatrix}
3.2
\end{pmatrix}$,
we have
\begin{align*}
\textstyle
\boldsymbol{\mathit{U}}^{(4)} &= 
\boldsymbol{\mathit{U}}^{(3)} + 
(\boldsymbol{\mathit{X}}^{(3, 1)})^\top
\boldsymbol{\mathit{X}}^{(3, 1)}
=
\boldsymbol{\mathit{U}}^{(3)} + 
\begin{pmatrix}
1 & 2.9 \\
2.9 & 8.41
\end{pmatrix}
,\\
\boldsymbol{\mathit{V}}^{(4)} &= 
\boldsymbol{\mathit{V}}^{(3)} + 
(\boldsymbol{\mathit{X}}^{(3, 1)})^\top
\boldsymbol{\mathit{Y}}^{(3, 1)}
=
\boldsymbol{\mathit{V}}^{(3)} + 
\begin{pmatrix}
3.2 \\
9.28
\end{pmatrix}
,\\
\phi_1^{(4)} &= 
(\boldsymbol{\mathit{U}}^{(4)} + \alpha\boldsymbol{\mathit{E}})^{-1}
\boldsymbol{\mathit{V}}^{(4)}
=
\begin{pmatrix}
5.56 \\
-0.87
\end{pmatrix}.
\end{align*}
\end{example}

\subsubsection{Incremental Learning Algorithm}
\label{sect-incremental-algorithm}

We revise Algorithm \ref{algorithm-learning} for incremental learning. 
For each $\mathit{t}_i\in\mathit{r}$, 
$\mathit{T}_i^{(\ell+\mathit{h})}\setminus\mathit{T}_i^{(\ell)}$ is retrieved in Formula \ref{equation-neighbor-increment}, 
rather than all the $\ell$ nearest neighbors in Line \ref{line-learning-lNN} in Algorithm \ref{algorithm-learning}.
Referring to Proposition \ref{the-incremental-bottom-up}, 
we incrementally compute
$\boldsymbol{\mathit{U}}^{(\ell+\mathit{h})}, \boldsymbol{\mathit{V}}^{(\ell+\mathit{h})}$ 
from 
$\boldsymbol{\mathit{U}}^{(\ell)}, \boldsymbol{\mathit{V}}^{(\ell)}$, 
together with 
$\boldsymbol{\mathit{Y}}^{(\ell, \Delta_{\mathit{h}})}$
and 
$\boldsymbol{\mathit{X}}^{(\ell, \Delta_{\mathit{h}})}$
on nearest neighbor increments. 
Finally, 
$\phi_i^{(\ell+\mathit{h})}$ is computed by Formula \ref{equation-mid-matrix-phi}.

It is worth noting that incrementally computing $\phi_i^{(\ell+\mathit{h})}$ only needs to cache $\boldsymbol{\mathit{U}}^{(\ell)}, \boldsymbol{\mathit{V}}^{(\ell)}$ in the previous step.
Earlier results such as $\boldsymbol{\mathit{U}}^{(\ell-\mathit{h})}, \boldsymbol{\mathit{V}}^{(\ell-\mathit{h})}$ could be discarded.
Given the same $\mathit{h}$, the incremental computation naturally supports stepping.

\subsubsection{Complexity Analysis}
\label{sect-incremental-complexity}

Table \ref{table-increment-time} lists the major steps and costs for learning  parameter $\phi_i^{(\ell+\mathit{h})}$ in Formula \ref{equation-mid-matrix-phi}. 
As shown, the costs of computing $\boldsymbol{\mathit{U}}$ and $\boldsymbol{\mathit{V}}$ from scratch using Formulas \ref{equation-mid-matrix-U} and \ref{equation-mid-matrix-V} are linear in terms of $\ell$.
With the incremental computation in Formulas \ref{equation-estimate-recursive-u} and \ref{equation-estimate-recursive-v} in Proposition \ref{the-incremental-bottom-up}, 
the costs become irrelevant to $\ell$. 
In other words, 
we reduce the learning cost from linear $O(m^2\ell+m^2h + m^3)$ to constant $O(m^2h + m^3)$ in terms of~$\ell$ tuples.

\begin{table}[t]
 \caption{Time complexity for learning parameter $\phi_i^{(\ell+\mathit{h})}$}
 \label{table-increment-time}
 \centering
 \begin{tabular}{rll}
 \hline\noalign{\smallskip} Computing  & From scratch  & Incremental  \\ \noalign{\smallskip}
 \hline\noalign{\smallskip}
 $\boldsymbol{\mathit{U}}$ & $m^2(\ell+h)$ &  $m^2h$\\ \noalign{\smallskip}
 $\boldsymbol{\mathit{V}}$ & $m(\ell+h)$  & $mh$  \\ \noalign{\smallskip}
 $(\boldsymbol{\mathit{U}})^{-1}$ & $m^3$ & $m^3$ \\ \noalign{\smallskip}
 $(\boldsymbol{\mathit{U}})^{-1}\boldsymbol{\mathit{V}}$ & $m^2$ &  $m^2$ \\ \noalign{\smallskip}
 \hline\noalign{\smallskip}
 \end{tabular}
 \end{table}

\section{Experiment}\label{sect-experiment}

While the theoretical analysis in Section \ref{sect-subsume} proves that our proposal subsumes some existing methods, 
the empirical evaluation particularly concerns 
how \textsf{IIM} outperforms the existing imputation approaches in practice, in Section \ref{experiment-imputation-comparision}.

\begin{table}[t]
 \caption{Dataset summary}
 \label{table-dataset}
 \centering
 \begin{tabular}{lllll}
 \hline\noalign{\smallskip} Dataset & $|\mathit{r}|$ & $|\mathcal{R}|$ & Source & Property \\ \noalign{\smallskip}
 \hline\noalign{\smallskip}
 ASF & 1.5k & 6 & UCI & no clear global regression \\ \noalign{\smallskip}
 CCS & 1k & 6 & UCI &  \\ \noalign{\smallskip}
 CCPP & 10k & 5 & UCI &  \\ \noalign{\smallskip}
 SN & 100k & 2 & UCI &  \\ \noalign{\smallskip}
 PHASE & 10k & 4 & Siemens & a clear global regression \\ \noalign{\smallskip}
 CA & 20k & 9 & KEEL & sparse with high dimension \\ \noalign{\smallskip}
 DA & 7k & 6 & KEEL & \\ \noalign{\smallskip}
 MAM & 1k & 5 & KEEL & real missing, no truth \\ \noalign{\smallskip}
 HEP & 200 & 19 & KEEL & real missing, no truth \\ \noalign{\smallskip}
 \hline\noalign{\smallskip}
 \end{tabular}
\end{table}

\subsection{Settings}

\subsubsection{Datasets}
\label{sect-dataset} 

We employ 9 datasets 
from different sources, 
UCI\footnote{\url{http://archive.ics.uci.edu/ml/datasets/}} \cite{lichman2013uci},
KEEL\footnote{\url{http://sci2s.ugr.es/keel/datasets.php}} \cite{DBLP:journals/mvl/Alcala-FdezFLDG11} and 
Siemens,
with various properties as summarized in Table \ref{table-dataset}. 
For instance, no clear linear regression is observed globally in the ASF dataset, i.e., with heterogeneity problem,
while the PHASE dataset has a clear regression relationship in three-phase electric power. 
The CA dataset involves 9 attributes with higher dimension, which leads to more serious sparsity issue.
The MAM and HEP datasets contain real-world missing values without ground truth, and are used for evaluating the classification application with / without imputation.

\subsubsection{Criteria}
\label{sect-Criteria}

Following the same line of evaluating data quality approaches \cite{DBLP:journals/pvldb/ArocenaGMMPS15},
for each dataset (except the two datasets without ground truth), 
we randomly select a set of tuples as $\{\mathit{t}_\mathit{x}\}$ by removing values on (multiple) attributes $\{\mathit{A}_\mathit{x}\}$ as missing values. 
The remaining tuples are considered as complete tuples in $\mathit{r}$.
When multiple incomplete attributes $\{\mathit{A}_\mathit{x}\}$ exist, we impute them one by one.
RMS error \cite{DBLP:conf/vldb/JefferyGF06} is employed to evaluate the imputation accuracy, 
$
\sqrt{
\frac{
\sum_{\mathit{t}_\mathit{x},\mathit{A}_\mathit{x}}
(\mathit{t}_\mathit{x}[\mathit{A}_\mathit{x}]-\mathit{t}_\mathit{x}'[\mathit{A}_\mathit{x}])^2
}{|\{(\mathit{t}_\mathit{x},\mathit{A}_\mathit{x})\}|}
},
$
where 
$\mathit{t}_\mathit{x}[\mathit{A}_\mathit{x}]$ is the original value (ground truth) of the incomplete attribute, 
and 
$\mathit{t}_\mathit{x}'[\mathit{A}_\mathit{x}]$ is the corresponding imputation. 
The lower the RMS error is, 
the better the imputation accuracy will be, i.e., closer to the truth.

The sparsity issue states that a tuple does not have sufficient neighbors that share the same/similar values. 
In other words, the truth value varies from the values suggested by complete neighbors.  
To evaluate the variance, we employ 
the \emph{coefficient of determination} 
\cite{devore2011probability},
$
\mathit{R}^2=1-\frac{
\sum_{\mathit{t}_\mathit{x}}(\mathit{t}_\mathit{x}[\mathit{A}_\mathit{m}] - \mathit{t}_\mathit{x}'[\mathit{A}_\mathit{m}])^2
}{
\sum_{\mathit{t}_\mathit{x}}(\mathit{t}_\mathit{x}[\mathit{A}_\mathit{m}] - \overline{\mathit{t}_\mathit{i}[\mathit{A}_\mathit{m}]})^2
}
,
$
where $\mathit{t}_i \in \mathit{r}$, 
$\mathit{t}_\mathit{x}[\mathit{A}_\mathit{m}]$ is the truth value, and 
$\mathit{t}_\mathit{x}'[\mathit{A}_\mathit{m}]$ is the value suggested by complete neighbors (e.g., by \textsf{kNN}).
We denote $\mathit{R}^2_S$ the $\mathit{R}^2$ measure on sparsity. 
The lower the measure $\mathit{R}^2_S$ is, 
the more serious the sparsity issue will be in the data. 

The heterogeneity issue states that tuples do not fit a single global model.
Similarly, we evaluate how the truth value varies from the values predicted by the single global model. 
Again, the aforesaid \emph{coefficient of determination} is employed, 
where 
$\mathit{t}_\mathit{x}'[\mathit{A}_\mathit{m}]$ is the value predicted by the single global model (e.g., by \textsf{GLR}).
The lower the measure $\mathit{R}^2_H$ is, 
the more serious the heterogeneity issue will be in the data.

\subsection{Comparison on Imputation Methods}
\label{experiment-imputation-comparision}

This experiment compares our proposal \textsf{IIM} with the existing approaches listed in Table~\ref{table-algorithm}  
in Section \ref{sect:preliminary}. 
We use the MICE implementation\footnote{https://github.com/stefvanbuuren/mice/tree/master/R} of \textsf{PMM} and \textsf{BLR} in R, 
the \textsf{XGB} implementation in R, 
and the existing \textsf{SVD} implementation\footnote{https://github.com/jeffwong/imputation}. 
Other approaches as well as our \textsf{IIM} are implemented in Java. 
Thereby, the corresponding time costs could be compared, e.g., in Figures 
\ref{exp:ca-completenum} and \ref{exp:asf-k}.
While some significantly worse results may not appear in the figures, 
the results of all methods can be found in Tables \ref{table-all-type} and \ref{table:performance}.

\begin{table*}[t]
 \caption{Imputation RMS error of \textsf{IIM} compared to the existing approaches listed in Table \ref{table-algorithm} over various datasets}
 \label{table-all-type}
 \centering
 \begin{tabular}{cccccccccccccccc}
 \hline\noalign{\smallskip}
Dataset & $\mathit{R}^2_S$ & $\mathit{R}^2_H$ & \textsf{IIM} & \textsf{kNN} &\textsf{kNNE} & \textsf{IFC} & \textsf{GMM} & \textsf{SVD} & \textsf{ILLS} & \textsf{GLR} & \textsf{LOESS} & \textsf{BLR} & \textsf{ERACER} & \textsf{PMM} & \textsf{XGB} \\ \noalign{\smallskip}
\hline\noalign{\smallskip}
ASF & 0.85 & 0.73 & \textbf{8.08} & 22.63 & 20.12 & 50.72 & 59.04 & 37.88 & 16.05 & 30.28 & 16.73 & 42.78 & 20.35 & 36.43 & 11.61 \\ \noalign{\smallskip}
CA & 0.03 & 0.90 & \textbf{0.49} & 2.02 & 1.85 & 2.03 & 2.12 & 50.11 & 12.76 & 0.6 & 0.54 & 0.88 & 0.6 & 0.77 & 0.7 \\ \noalign{\smallskip}
CCPP & 0.95 & 0.93 & \textbf{3.75} & 3.98 & 4.13 & 14.08 & 23.09 & 6.79 & 5.78 & 4.58 & 4.25 & 6.55 & 3.97 & 6.19 & 4.45 \\ \noalign{\smallskip}
CCS & 0.63 & 0.56 & \textbf{10.45} & 12.84 & 11.13 & 21.39 & 24.95 & 25.59 & 13.67 & 13.64 & 12.76 & 20.51 & 11.25 & 18.85 & 11.26 \\ \noalign{\smallskip}
DA & 0.65	& 0.68 & \textbf{15.52} & 16.99 & 17.75 & 22.92 & 23.99 & 21.92 & 94.5 & 16.68 & 15.88 & 23.69 & 16.18 & 23.47 & \textbf{15.56} \\ \noalign{\smallskip}
PHASE & 0.9 & 0.91 & \textbf{3.31} & 3.51 & 3.42 & 5.41 & 11.35 & 5.28 & 3.59 & \textbf{3.32} & \textbf{3.32} & 4.73 & \textbf{3.32} & 4.64 & \textbf{3.36} \\ \noalign{\smallskip}
SN & 0.79	& 0.05 & \textbf{0.11} & \textbf{0.12} & \textbf{0.12} & 0.28 & 0.43 & -  & - & 0.27 & 0.20 & 0.4 & \textbf{0.13} & 0.28 & - \\ \noalign{\smallskip}
 \hline\noalign{\smallskip}
 \end{tabular}
\end{table*}

\subsubsection{Imputation on Various Datasets}
\label{sect-experiment-performance-datasets}

For each dataset in Table \ref{table-all-type}, 
we randomly pick $5\%$ tuples as $\mathit{t}_\mathit{x}$ with one missing value on a random attribute $\mathit{A}_\mathit{x}$.
That is, there are $5\%\frac{1}{|\mathcal{R}|}$ missing values w.r.t.\ the total values in each dataset, where $|\mathcal{R}|$ is the number of attributes in the dataset.
For instance, the CCPP data with 5 attributes has  $\frac{5\%}{5}=1\%$ missing values.

When a dataset is with high sparsity but low heterogeneity, i.e., small $\mathit{R}^2_S$ but large $\mathit{R}^2_H$, such as CA in Table \ref{table-all-type}, 
the \textsf{GLR} approach using the predicted value via the regression model 
shows a better imputation performance (RMS=0.6) than 
the \textsf{kNN} method using the (aggregated) value in the complete neighbor tuples (RMS=2.02).

Nevertheless, 
our proposed \textsf{IIM} always shows the lowest imputation error. 
The result is not surprising referring to the theoretical analysis in Section \ref{sect-subsume} that our proposal subsumes \textsf{GLR} and \textsf{kNN} as special cases.

To show applicability, 
we report the results on the larger dataset SN 
in Table \ref{table-all-type}. 
As shown, 
the better imputation result of our proposed \textsf{IIM} is still consistently observed. 
(The results of \textsf{SVD}, \textsf{ILLS} and \textsf{XGB} are not available since they cannot be implemented on only two attributes.)

\begin{table*}[t]
 \caption{Imputation RMS error on various incomplete attribute $\mathit{A}_\mathit{x}$ over ASF dataset with 100 incomplete tuples}
 \label{table:performance}
 \centering
 \begin{tabular}{cccccccccccccccc}
\hline\noalign{\smallskip} 
 & $\mathit{R}^2_S$ & $\mathit{R}^2_H$ & \textsf{IIM} & \textsf{kNN} & \textsf{kNNE} & \textsf{IFC} & \textsf{GMM} & \textsf{SVD} & \textsf{ILLS} & \textsf{GLR} & \textsf{LOESS} & \textsf{BLR} & \textsf{ERACER} & \textsf{PMM} & \textsf{XGB} \\ \noalign{\smallskip}
 \hline\noalign{\smallskip}
$A_1$ & 0.47 & 0.46 & \textbf{192.5} & 235.2 & 247.8 & 326.9 & 334.8 & 320.4 & 248.3 & 234.4 & 201.8 & 328.5 & 206.8 & 289.9 & 204.9 \\ \noalign{\smallskip}
$A_2$ & 0.85 & 0.73 & \textbf{8.08} & 22.63 & 20.12 & 50.72 & 59.04 & 37.88 & 16.05 & 30.28 & 16.73 & 42.78 & 20.35 & 36.43 & 11.61 \\ \noalign{\smallskip}
$A_3$ & 0.73 & 0.5 & \textbf{1.49} & 5.11 & 4.08 & 8.87 & 12.15 & 9.67 & 4.73 & 6.54 & 3.66 & 9.18 & 4.51 & 8.72 & 2.07 \\ \noalign{\smallskip}
$A_4$ & 0.03 & 0.12 & \textbf{12.82} & 15.74 & 13.28 & 15.65 & 16.74 & 15.16 & 17.62 & 14.68 & 13.84 & 21.14 & 14.68 & 20.23 & 13.24 \\ \noalign{\smallskip}
$A_5$ & 0.79 & 0.63 & \textbf{13.85} & 64.94 & 60.29 & 125.58 & 138.22 & 88.76 & 34.9 & 80.54 & 55.95 & 116.65 & 58.01 & 90.53 & 23.23 \\ \noalign{\smallskip}
$A_6$ & 0.78 & 0.51 & \textbf{3.22} & 3.28 & 4.59 & 6.39 & 7.39 & 45.25 & 11.82 & 4.8 & 3.4 & 7.02 & \textbf{2.92} & 6.29 & 15.25 \\ \noalign{\smallskip}
 \hline\noalign{\smallskip}
 \end{tabular}
\end{table*}

\subsubsection{Varying the Missing Attribute $\mathit{A}_\mathit{x}$}
\label{sect-experiment-missing-attribute}

Table \ref{table:performance} reports the results on various incomplete attributes $\mathit{A}_\mathit{x}$ over the ASF data. 
Owing to the different ranges of domain values on various attributes, 
the imputation RMS error differs in attributes.

Approaches perform variously over the attributes with different domain characteristics in terms of sparsity and heterogeneity. 
In Table \ref{table:performance}, 
for attribute $\mathit{A}_4$ with 
small $\mathit{R}^2_S$ (high sparsity) but 
large $\mathit{R}^2_H$ (low heterogeneity),  
the attribute model methods (\textsf{GLR} and \textsf{LOESS} using the values predicted by regression models) perform better than the tuple model methods (\textsf{kNN}  using the aggregated value of complete neighbor tuples).
In contrast, 
for attribute $\mathit{A}_6$ with 
large $\mathit{R}^2_S$ (low sparsity) but small $\mathit{R}^2_H$ (high heterogeneity), 
\textsf{kNN} outperforms \textsf{GLR}.
Nevertheless, 
since our proposal  concerns both sparsity and heterogeneity, 
\textsf{IIM} consistently shows the best performance. 
The results verify the superiority of our proposal.

\begin{figure}[t]\centering
\begin{minipage}{\expwidths}\centering
\hspace{-0.5em}%
\includegraphics[width=\expwidths]{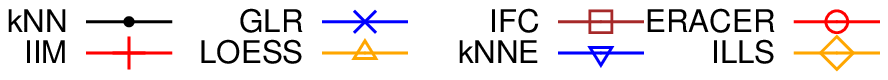} 
\hspace{-0.5em}%
\includegraphics[width=0.5\expwidths]{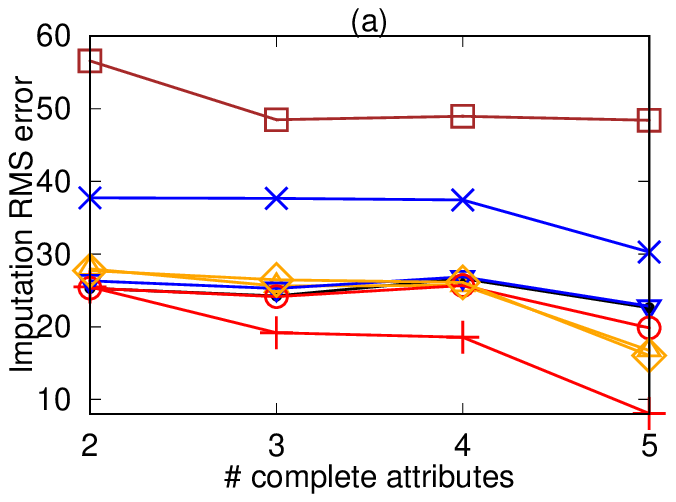}%
\hspace{-0.5em}%
\includegraphics[width=0.5\expwidths]{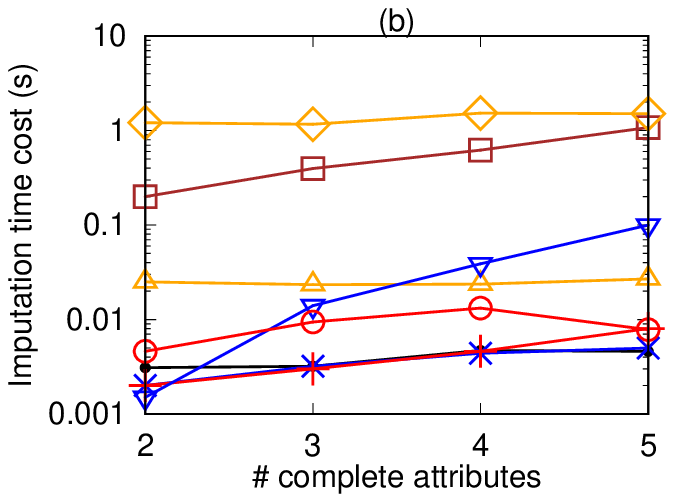}%
\end{minipage}
\caption{Varying the number of complete attributes $|\mathcal{F}|$, 
over ASF with 100 incomplete tuples}
\label{exp:asf-completeAttrNum}
\end{figure}

\begin{figure}[t]\centering
\begin{minipage}{\expwidths}\centering
\hspace{-0.5em}%
\includegraphics[width=\expwidths]{exp-label-hor} 
\hspace{-0.5em}%
\includegraphics[width=0.5\expwidths]{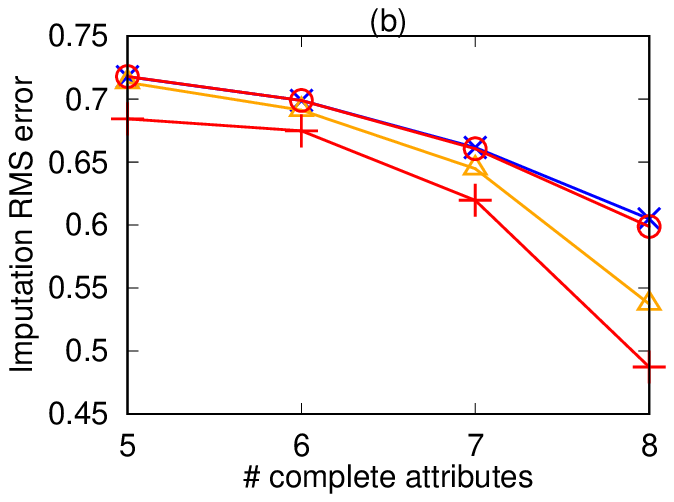}%
\hspace{-0.5em}%
\includegraphics[width=0.5\expwidths]{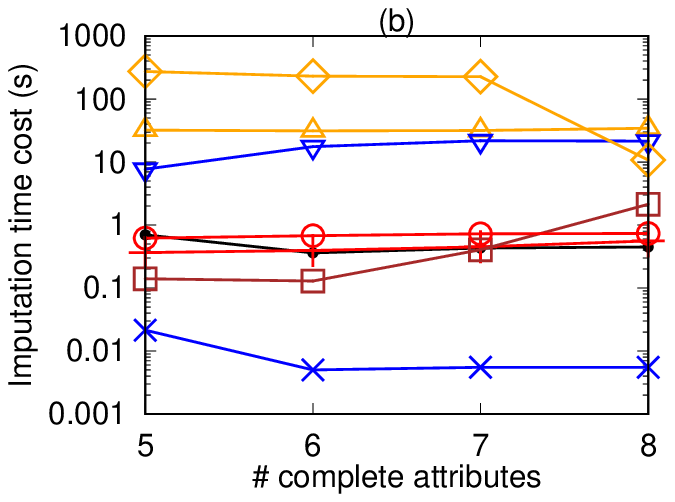}%
\end{minipage}
\caption{Varying the number of complete attributes $|\mathcal{F}|$, 
over CA with 1k incomplete tuples}
\label{exp:ca-completeAttrNum}
\end{figure}

\subsubsection{Varying the Number of Complete Attributes $|\mathcal{F}|$}
\label{sect-experiment-varying-complete-attributes}

When preparing the datasets, 
we randomly pick a certain percent (\%) tuples as $\mathit{t}_\mathit{x}$ with one missing value on a random attribute $\mathit{A}_\mathit{x}$.
By default, 
all the remaining attributes are used as complete neighbors for imputation, 
i.e., $\mathcal{F}=\mathcal{R}\setminus\{\mathit{A}_\mathit{x}\}$. 
In order to evaluate the imputation with different sizes of complete attributes, 
the experiments in 
Figures \ref{exp:asf-completeAttrNum} 
and \ref{exp:ca-completeAttrNum} 
consider a subset of $\mathcal{R}\setminus\{\mathit{A}_\mathit{x}\}$ as the complete attributes $\mathcal{F}$. 
For instance, a number of complete attributes $|\mathcal{F}|=2$ in the x axis denotes 
$\mathcal{F}=\{\mathit{A}_1,\mathit{A}_2\}$, 
instead of considering all the attributes in $\mathcal{R}\setminus\{\mathit{A}_\mathit{x}\}=\{\mathit{A}_1,\mathit{A}_2,\mathit{A}_3,\dots\}$ as complete attributes.

Figures \ref{exp:asf-completeAttrNum} 
and \ref{exp:ca-completeAttrNum} 
present the results on various number of complete attributes $|\mathcal{F}|$. 
For most approaches, it is not surprising that imputation improves under more complete attributes. 
Specifically, with more attributes in $\mathcal{F}$, the regression from $\mathcal{F}$ to $\mathit{A}_\mathit{x}$ will be more reliable (if exists). 
Furthermore, the neighbors found w.r.t.\ larger $\mathcal{F}$ are more likely to share values. 
With both aforesaid benefits, our \textsf{IIM} shows  more significant improvements when $\mathcal{F}$ is large.

Figures \ref{exp:asf-completeAttrNum}(b) and \ref{exp:ca-completeAttrNum}(b) report the time cost of \textsf{IIM} in the imputation phase 
(the offline learning phase only needs to be processed once for imputing different incomplete tuples).
In contrast, \textsf{LOESS} and \textsf{ILLS} need to online learn the local regression over the neighbors of the input incomplete tuple, 
and thus have high imputation time cost.
It is not surprising that \textsf{IIM} shows similar time cost as \textsf{kNN},
since both approaches need to find $k$ nearest neighbors.

\begin{figure}[t]\centering
\begin{minipage}{\expwidths}\centering
\hspace{-0.5em}%
\includegraphics[width=\expwidths]{exp-label-hor} 
\hspace{-0.5em}%
\includegraphics[width=0.5\expwidths]{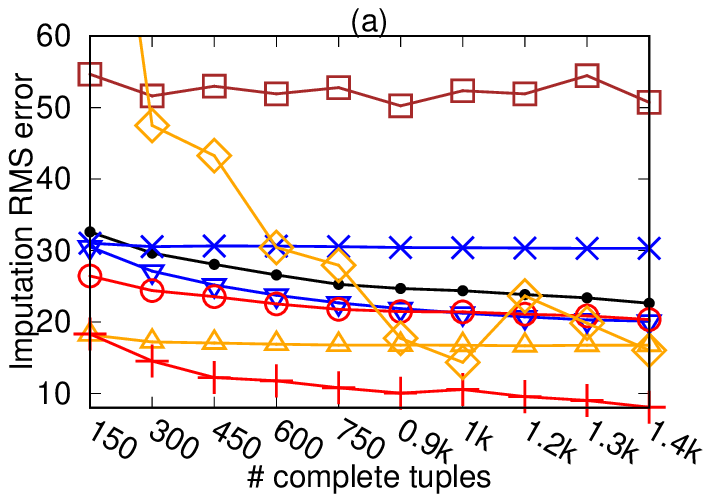}%
\hspace{-0.5em}%
\includegraphics[width=0.5\expwidths]{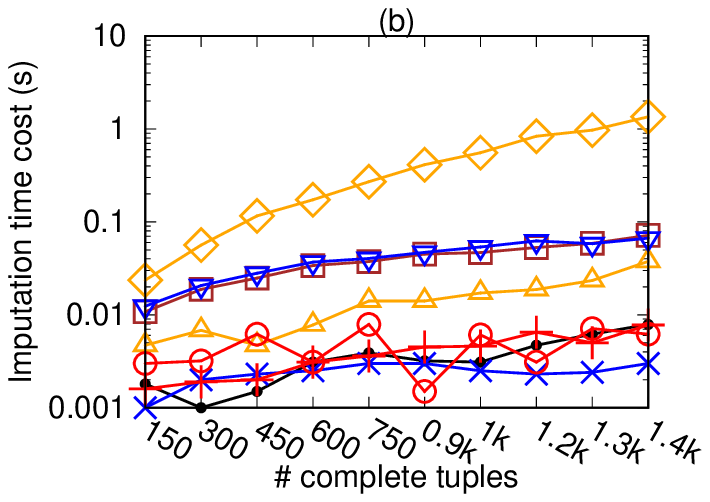}%
\end{minipage}
\caption{Varying the number of complete tuples $\mathit{n}=|\mathit{r}|$, over ASF with 100 incomplete tuples}
\label{exp:asf-completenum}
\end{figure}

\begin{figure}[t]\centering
\begin{minipage}{\expwidths}\centering
\hspace{-0.5em}%
\includegraphics[width=\expwidths]{exp-label-hor} 
\hspace{-0.5em}%
\includegraphics[width=0.5\expwidths]{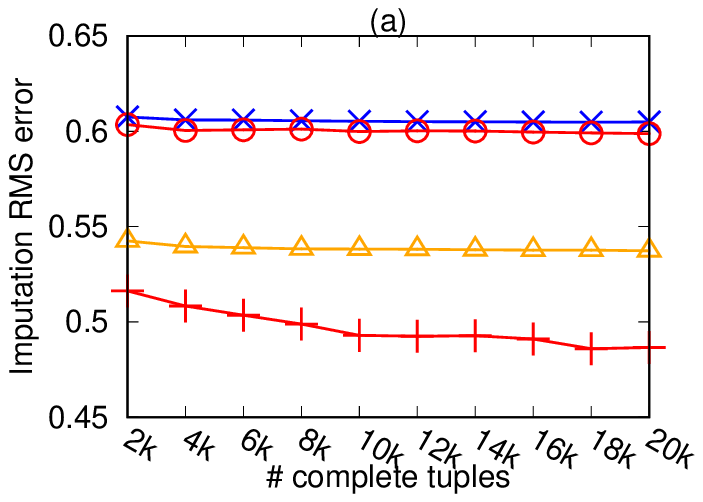}%
\hspace{-0.5em}%
\includegraphics[width=0.5\expwidths]{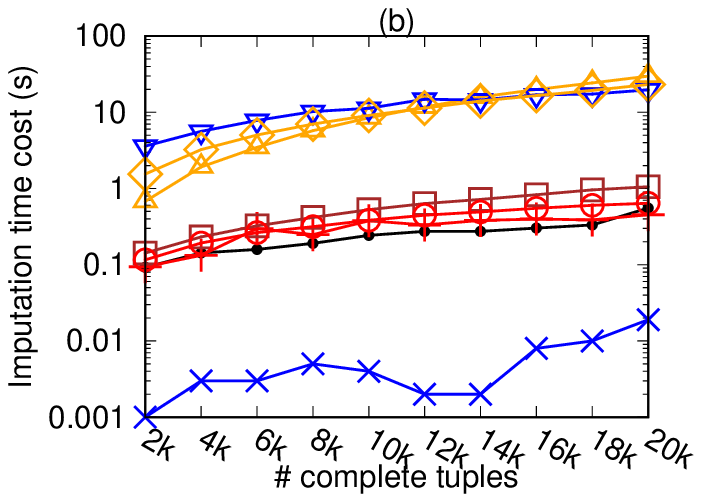}%
\end{minipage}
\caption{Varying the number of complete tuples $\mathit{n}=|\mathit{r}|$, over CA with 1k incomplete tuples}
\label{exp:ca-completenum}
\end{figure}

\subsubsection{Varying the Number of Complete Tuples $\mathit{n}=|\mathit{r}|$}

Figures \ref{exp:asf-completenum} and \ref{exp:ca-completenum} report the results 
by randomly selecting $\mathit{n}$ tuples from the dataset as $\mathit{r}$ of complete tuples. 
Generally, more complete tuples lead to better imputation performance. 
The interesting result in Figure \ref{exp:asf-completenum}(a) is that \textsf{kNN} relies more on complete tuples to achieve lower imputation error, since it requires the presence of sufficient neighbors sharing  similar values. 
Our \textsf{IIM} utilizing the individual regressions of tuples benefits from more complete tuples as well.

\begin{figure}
\begin{minipage}{\expwidths}\centering
\hspace{-0.5em}%
\includegraphics[width=\expwidths]{exp-label-hor} 
\hspace{-0.5em}%
\includegraphics[width=0.5\expwidths]{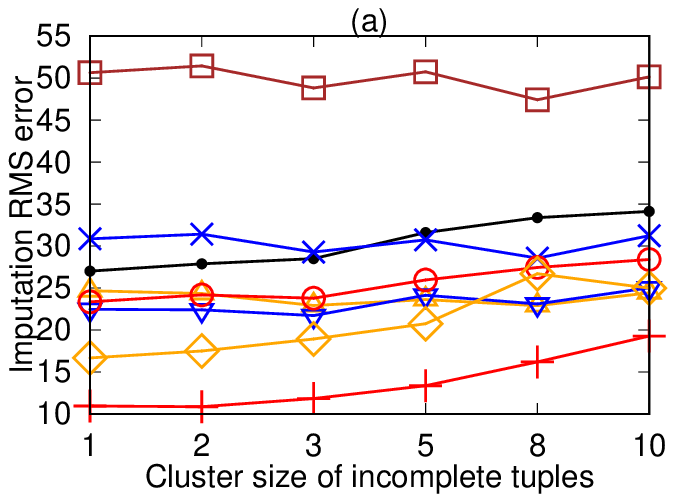}%
\hspace{-0.5em}%
\includegraphics[width=0.5\expwidths]{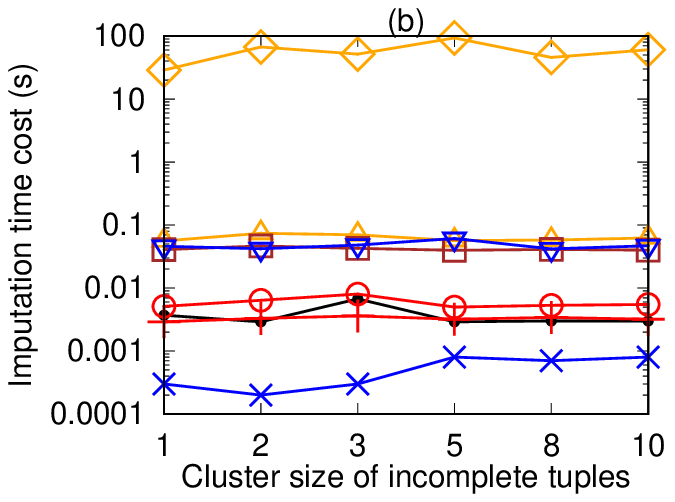}%
\end{minipage}
\caption{Varying the cluster size of incomplete tuples, over ASF with 100 incomplete tuples in total}
\label{exp:asf-mk}
\end{figure}

\subsubsection{Varying the Cluster Size of Incomplete Tuples}
\label{sect-experiment-cluster-incomplete}

Rather than introducing missing values in random tuples,  
we consider incomplete tuples that cluster together.
That is, complete neighbors are very far away.
Figure \ref{exp:asf-mk} reports the results under various sizes of  incomplete tuple clusters.
For example, a cluster size 3 denotes that the 2 closest neighbors are also incomplete tuples. 
It is not surprising that 
with the increase of incomplete tuple cluster size,
all the tuple model based imputation methods relying on the closest neighbors (e.g., \textsf{kNN}, \textsf{ILLS}) become worse.
On the other hand, the attribute model based methods (such as \textsf{GLR} or \textsf{LOESS}) are relatively stable.
Again, our proposed \textsf{IIM} still shows the best performance, 
since it does not rely on the neighbor tuples to share the same values, 
and thus can cope with the sparsity issue introduced by the clusters of incomplete tuples.

\subsection{Evaluation on Individual Learning}
\label{experiment-individual}

In this section, we evaluate the characteristic of proposed techniques on the following aspects to show the performance and rational behind \textsf{IIM}.

\subsubsection{Varying the Number of Imputation Neighbors $\mathit{k}$}
\label{experiment-imputation-number}

\begin{figure}[t]\centering
\begin{minipage}{\expwidths}\centering
\hspace{-0.5em}%
\includegraphics[width=\expwidths]{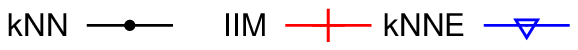}
\hspace{-0.5em}%
\includegraphics[width=0.5\expwidths]{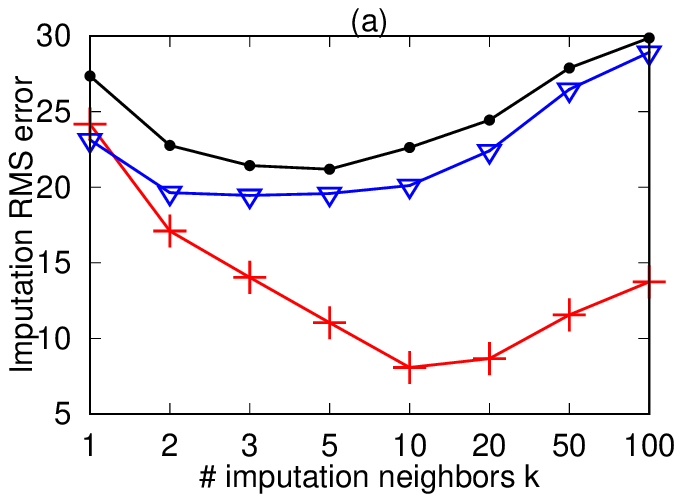}%
\hspace{-0.5em}%
\includegraphics[width=0.5\expwidths]{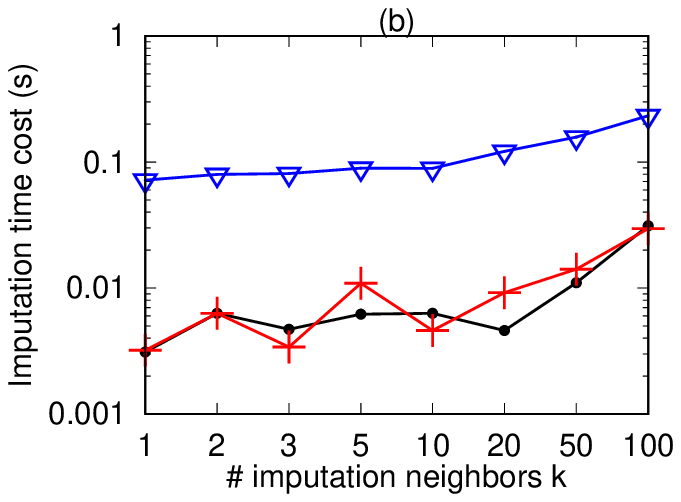}%
\end{minipage}
\caption{Varying the number of imputation neighbors $\mathit{k}$, over ASF with 100 incomplete tuples}
\label{exp:asf-k}
\end{figure}

\begin{figure}[t]\centering
\begin{minipage}{\expwidths}\centering
\hspace{-0.5em}%
\includegraphics[width=\expwidths]{exp-label-k} 
\hspace{-0.5em}%
\includegraphics[width=0.5\expwidths]{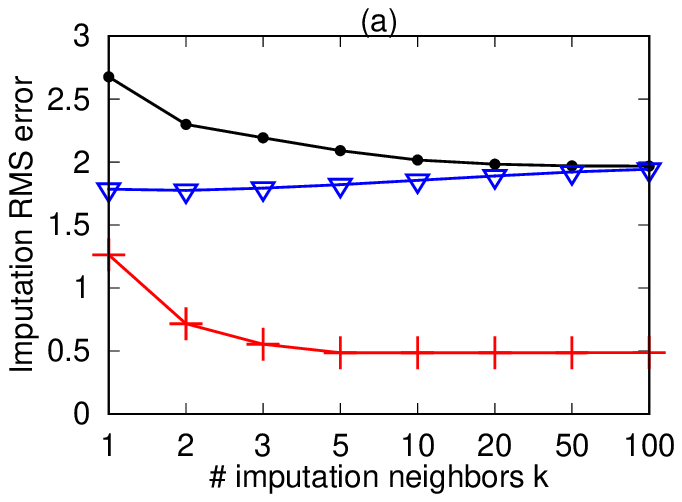}%
\hspace{-0.5em}%
\includegraphics[width=0.5\expwidths]{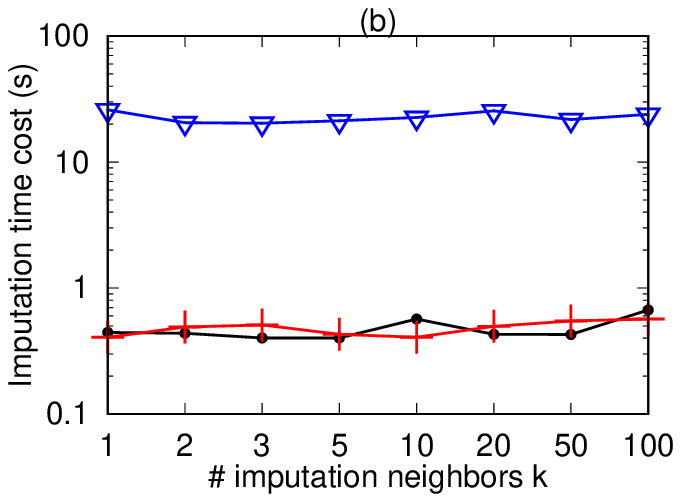}%
\end{minipage}
\caption{Varying the number of imputation neighbors $\mathit{k}$, over CA with 1k incomplete tuples}
\label{exp:ca-k}
\end{figure}

This experiment evaluates various number of imputation neighbors $\mathit{k}$. 
It is used in both \textsf{kNN}, \textsf{kNNE} and our \textsf{IIM} (in Algorithm \ref{algorithm-imputation} of imputation phase). 
Figures \ref{exp:asf-k} and \ref{exp:ca-k} report the results on ASF (having heterogeneity issues) and CA (having sparsity property) with $5\%$ incomplete tuples. 
Generally, a moderately large $\mathit{k}$ is preferred. 
If $\mathit{k}$ is too small, it is not reliable to support the imputation. 
On the other hand, if $\mathit{k}$ is too large, irrelevant tuples may distract the imputation, 
as illustrated in Figure \ref{exp:asf-k}(a). 
For the CA data with sparsity issue
in Figure \ref{exp:ca-k}(a), 
changing the number of neighbors $\mathit{k}$ does not help much in imputation. 
(Some significantly worse results do not appear in the figure, such as \textsf{kNN} as shown in Table \ref{table-all-type}.)

\begin{figure}[t]\centering
\begin{minipage}{\expwidths}\centering
\hspace{-0.5em}%
\includegraphics[width=0.9\expwidths]{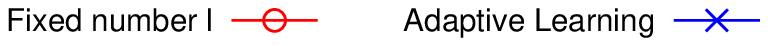} 
\hspace{-0.5em}%
\includegraphics[width=0.5\expwidths]{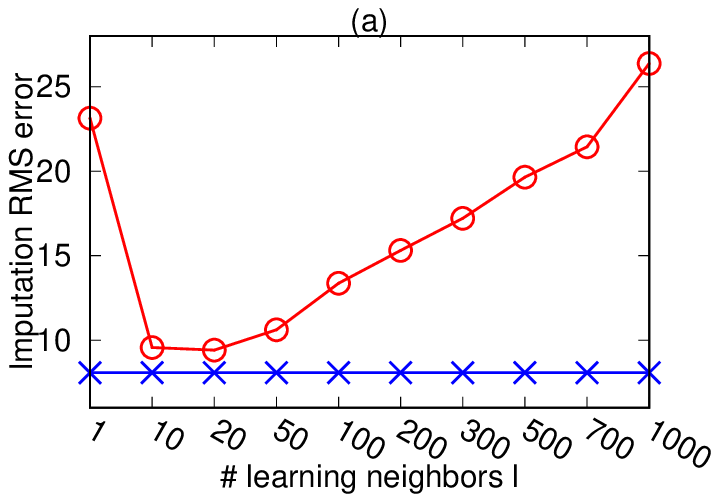}%
\hspace{-0.5em}%
\includegraphics[width=0.5\expwidths]{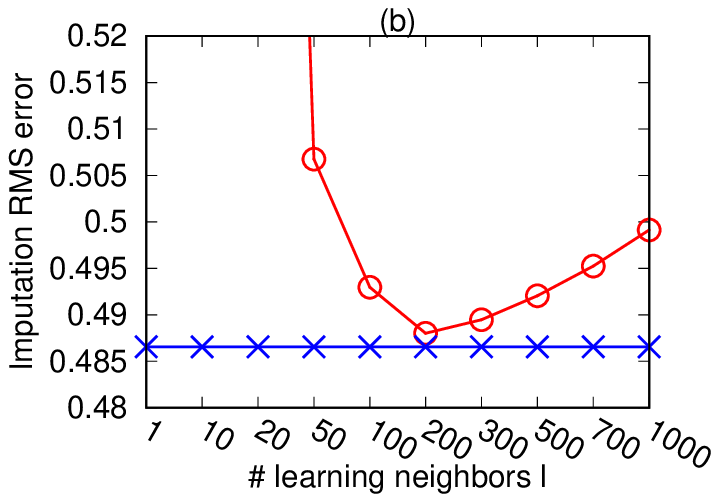}%
\end{minipage}
\caption{Comparison between adaptive learning and the learning over various fixed number $\ell$ of learning neighbors, over (a) ASF and (b) CA}
\label{exp:adaptive-fixl}
\end{figure}

\subsubsection{Evaluating Adaptive Learning}\label{sect-experiment-adaptive}

This experiment evaluates two aspects: 
(1) how the fixed number $\ell$ of learning neighbors for all tuples in Algorithm \ref{algorithm-learning} affects the imputation results; 
and 
(2) does the adaptive learning with distinct number of learning neighbors for different tuple in Algorithm \ref{algorithm-adaptive} truly improve the imputation? 

First, 
as shown in Figure \ref{exp:adaptive-fixl}, 
a small number $\ell$ of learning neighbors may suffer from the overfitting problem and lead to poor imputation.
On the other hand, when $\ell$ is too large, 
the learned individual model may suffer from the heterogeneity problem (under-fitting) and hence also has bad performance.
Manually choosing a proper $\ell$ is non-trivial, 
which is very different from datasets as illustrated in Figures \ref{exp:adaptive-fixl} (a) and (b).

Nevertheless, 
the proposed Adaptive Learning Algorithm \ref{algorithm-adaptive} can successfully address this problem, 
by adaptively considering a distinct number $\ell$ of learning neighbors for each tuple individually.
As illustrated in Figure \ref{exp:adaptive-fixl}, the performance of adaptive learning is better than setting a fixed $\ell$ for all tuples.

\begin{figure}[t]\centering
\begin{minipage}{\expwidths}\centering
\hspace{-0.5em}%
\includegraphics[width=0.9\expwidths]{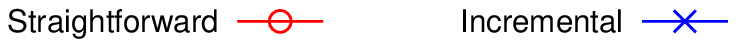} 
\hspace{-0.5em}%
\includegraphics[width=0.5\expwidths]{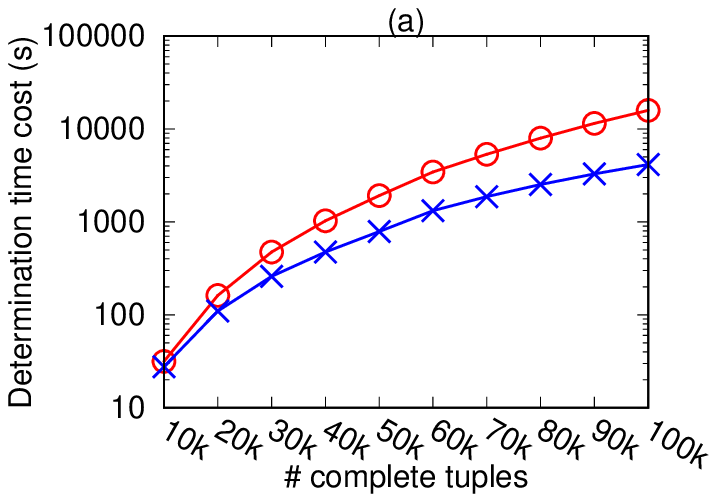}%
\hspace{-0.5em}%
\includegraphics[width=0.5\expwidths]{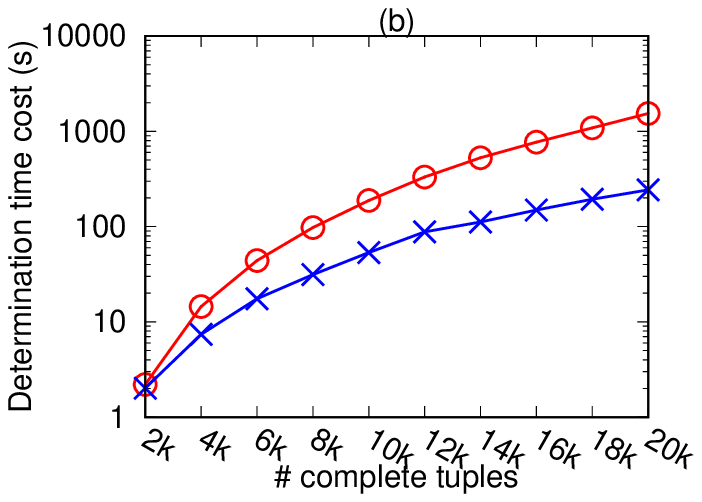}%
\end{minipage}
\caption{Scalability of adaptive learning (with straightforward and incremental computation) on the number $\mathit{n}$ of tuples in $\mathit{r}$, over (a)  SN and 
(b)  CA }
\label{exp:determineL-time}
\end{figure}

\subsubsection{Evaluating Incremental Learning}
\label{sect-experiment-incremental}

Figure \ref{exp:determineL-time} reports the time cost of adaptive learning using straightforward and incremental computation
(with stepping $\mathit{h}=50$) under various number $\mathit{n}$ of tuples in $\mathit{r}$.
The incremental learning algorithm devised in Section \ref{sect-incremental-algorithm} 
shows up to one order of magnitude improvement compared to 
the straightforward adaptive learning Algorithm \ref{algorithm-adaptive}. 
The result is not surprising, 
since the incremental computation reduces the time cost of parameter learning from linear to constant (in terms of $\ell$), 
as shown in Table \ref{table-increment-time} in Section \ref{sect-incremental-complexity}.
To show scalability, 
we report time cost of adaptive learning on SN in Figure \ref{exp:determineL-time}(a).
Again, the result is generally similar to the CA dataset with 20k tuples.

\begin{figure}[t]\centering
\begin{minipage}{\expwidths}\centering
\hspace{-0.5em}%
\includegraphics[width=0.9\expwidths]{exp-label-inc} 
\hspace{-0.5em}%
\includegraphics[width=0.5\expwidths]{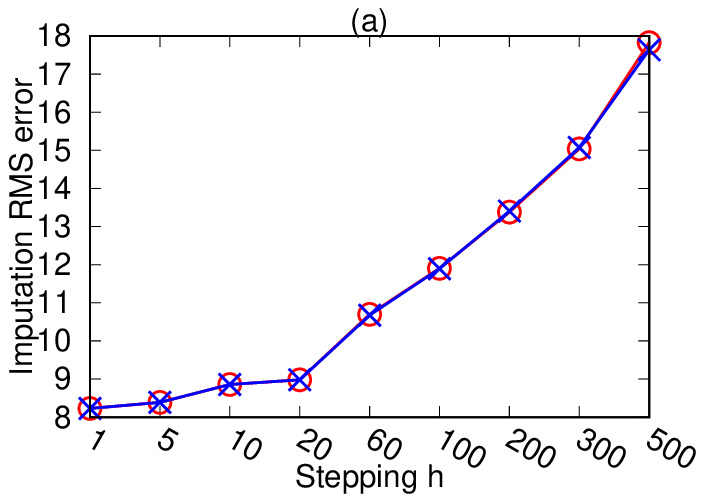}%
\hspace{-0.5em}%
\includegraphics[width=0.5\expwidths]{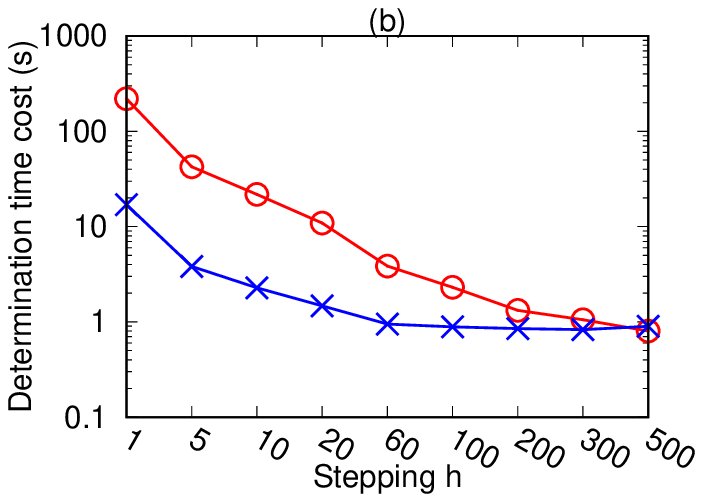}%
\end{minipage}
\caption{Varying stepping $\mathit{h}$ over ASF}
\label{exp:asf-stepping}
\end{figure}

\subsubsection{Tradeoff via Stepping}
\label{sect-experiment-stepping}

Figure \ref{exp:asf-stepping} 
present the results on varying the stepping $\mathit{h}$ studied in Section \ref{sect-stepping}. 
The smaller the $\mathit{h}$ is, the more the candidate $\ell$ values are considered.
When $\mathit{h}=1$, all the possible $\ell$ values are evaluated.
It is not surprising that 
a small stepping $\mathit{h}$ with more candidate $\ell$ values considered leads to lower imputation error in Figure \ref{exp:asf-stepping}(a), 
while the corresponding time cost is higher in Figure \ref{exp:asf-stepping}(b).
The exactly same imputation errors of straightforward and incremental determination algorithms verify the correctness of incremental computation. 
Figure \ref{exp:asf-stepping}(b) demonstrates again the significant improvement in time cost by the incremental determination algorithm.

\subsection{Applications with Imputation}
\label{sect-application}

\begin{table*}[t]
 \caption{Clustering purity on ASF \& CA, and Classification f1-score on MAM \& HEP with real missing values}
 \label{table:application}
 \centering
 \begin{tabular}{cccccccccccccccccc}
 \hline\noalign{\smallskip} & Missing & \textsf{IIM} & \textsf{Mean} & \textsf{kNN} & \textsf{kNNE} & \textsf{IFC} & \textsf{GMM} & \textsf{SVD} & \textsf{ILLS} & \textsf{GLR} & \textsf{LOESS} & \textsf{BLR} & \textsf{ERACER} & \textsf{PMM} & \textsf{XGB} \\ \noalign{\smallskip}
 \hline\noalign{\smallskip}
ASF & 0.697 & \textbf{0.916} & 0.883 & 0.898 & 0.901 & 0.891 & 0.883 & 0.737 & 0.799 & 0.902 & 0.908 & 0.876 & 0.912 & 0.898 & 0.904 \\ \noalign{\smallskip}
CA & 0.652 & \textbf{0.817} & 0.722 & 0.699 & 0.735 & 0.693 & 0.727 & 0.632 & 0.642 & 0.752 & 0.786 & 0.719 & 0.753 & 0.796 & 0.665 \\ \noalign{\smallskip}
 \hline\noalign{\smallskip}
MAM & 0.822 & \textbf{0.828} & 0.822 & 0.82 & 0.819 & 0.82 & 0.816 & 0.818 & 0.817 & 0.819 & 0.817 & 0.823 & 0.818 & 0.814 & 0.825 \\ \noalign{\smallskip}
HEP & 0.847 & \textbf{0.865} & 0.823 & 0.845 & 0.839 & 0.819 & 0.839 & 0.832 & 0.845 & 0.819 & 0.845 & 0.826 & 0.813 & 0.832 & 0.826 \\ \noalign{\smallskip}
 \hline\noalign{\smallskip}
 \end{tabular}
\end{table*}

It is known that dirty data may seriously mislead applications like clustering \cite{DBLP:conf/kdd/SongLZ15}.
To demonstrate the effectiveness of imputation in real applications, we consider the clustering and classification tasks over the data with and without missing data imputation, 
employing the \textsf{kmeans} and \textsf{ibk} (kNN classifier) implementations provided by Weka\footnote{\url{http://www.cs.waikato.ac.nz/ml/weka/}}, respectively.

\subsubsection{Clustering Application}
\label{sect-application-clustering}

The clustering algorithm is first performed over the original dataset without missing values. 
The returned cluster labels are served as ground truth. 
Missing values are then randomly introduced in the dataset as described in Section \ref{sect-Criteria}.
We apply various approaches to impute the missing values. 
The clustering algorithm is conducted again over the data with missing values and the imputed dataset.

We evaluate the accuracy of the clustering results over the data with/without imputation, by comparing to the truth clusters obtained from the original complete data. 
The purity \cite{DBLP:books/mk/HanKP2011} measure is employed, 
which counts for each cluster the number of data points from the most common class (truth cluster).
The higher the purity is, the better the imputation improves clustering.

The first two lines of Table \ref{table:application} show the clustering results.
The results are generally analogous to the imputation errors in 
Table \ref{table-all-type}.
Methods with lower imputation error lead to higher clustering purity.
The highest clustering accuracy by \textsf{IIM} demonstrates again the superiority of our proposal.

We also report the clustering results by simply discarding the incomplete tuples with missing values. 
As shown in the first column,
the clustering accuracy of the remaining complete tuples after discarding  significantly drops. 
The reason is that by removing many incomplete tuples, the dataset becomes even more incomplete and fails to form accurate clusters.
The results verify the motivation of imputing missing values rather than simply discarding incomplete tuples.

\subsubsection{Classification with Real-World Missing Values}

The classification application is evaluated on the MAM, 
and HEP datasets, 
where each tuple is labeled with classes and real-world missing values are naturally embedded without ground truth.
We use 5-fold cross validation,
where missing values exist both in training and testing sets.
The last three lines in Table \ref{table:application} report 
the f1-score accuracy of classification with and without imputation.
Our \textsf{IIM} with better imputation performance (in the previous experiments) shows again more significant improvement of classification accuracy.

\section{Conclusions}
\label{sect:conclusion}

To cope with the challenges of sparsity (no sufficient similar neighbors) and heterogeneity (tuples do not fit the same regression model)
in imputing numerical data,
we propose \textsf{IIM}, Imputation via Individual Models.
The rationale of our proposal is illustrated first by theoretically proving that 
some existing approaches are indeed special cases of \textsf{IIM} under extreme settings 
(i.e., $\ell=1$ or $\ell=\mathit{n}$ in Propositions \ref{the:l0-equal} and \ref{the-LR-equal}). 
It further motivates us to select a proper number $\ell$ of learning neighbors (in between the extreme $1$ and $\mathit{n}$) to avoid over-fitting or under-fitting. 
Again, owing to the heterogeneity issue, 
the number $\ell$ of learning neighbors could be different for learning the individual models of different tuples.
Through a validation step, 
we adaptively determine a model for each complete tuple 
that can best impute other tuples (in validation).
Efficient incremental computation is devised for adaptive learning,  
where the time complexity of learning a model reduces from linear to constant.
Experiments on read data 
demonstrate the superiority of our proposal. 

Future studies may further consider to answer queries directly over multiple imputation candidates suggested by different individual models \cite{DBLP:conf/sigmod/LianCS10}, 
rather determining exactly one imputation.    
Moreover, instead of study the regression models directly on data values, one may further investigate the dependency models \cite{DBLP:journals/vldb/Song0Y13} on the distances of values \cite{DBLP:journals/isci/SongZ014} for imputation. 

\subsubsection*{Acknowledgement}
This work is supported 
by
National Key Research Program of China under Grant 2016YFB1001101;
China NSFC under Grants 61572272 and 71690231.

\bibliographystyle{abbrv}
\bibliography{missing}

\vspace{-4em}
\includegraphics[height=0in]{tail}
\end{document}